\documentclass[cleveref,thm-restate,UKenglish,a4paper,pdfa]{lipics-v2021}
\nolinenumbers

\EventEditors{Philip Bille, Seth Pettie, and Sabine Storandt}
\EventNoEds{3}
\EventLongTitle{34th Annual European Symposium on Algorithms (ESA 2026)}
\EventShortTitle{ESA 2026}
\EventAcronym{ESA}
\EventYear{2026}
\EventDate{August 31--September 4, 2026}
\EventLocation{L'Aquila, Italy}
\EventLogo{}
\SeriesVolume{388}
\ArticleNo{20}

\newif\ifarxiv
\arxivtrue

\usepackage{newunicodechar}
\usepackage{silence}
\newunicodechar{♢}{\tikz \node[inner sep=1.5,draw,diamond] {};}
\newunicodechar{☆}{\tikz \node[inner sep=1,draw,star,star point ratio=2] {};}
\newunicodechar{△}{\triangle}
\newunicodechar{⬜}{\kern 0.5pt\tikz \node[inner sep=1.7,draw,regular polygon,regular polygon sides=4] {};\kern 0.5pt}
\newunicodechar{◯}{\tikz[baseline=-3pt] \node[inner sep=1.7,draw,cloud,cloud puffs=4,cloud puff arc=190] {};}

\newunicodechar{⊥}{\bot}
\newunicodechar{•}{\item}
\newunicodechar{✓}{\checkmark}
\newunicodechar{✗}{\xmark}\def\xmark{\ding{55}}
\newunicodechar{…}{\dots}
\newunicodechar{≔}{\coloneqq}
\newunicodechar{⁻}{^-}
\newunicodechar{⁺}{^+}
\newunicodechar{₋}{_-}
\newunicodechar{₊}{_+}
\newunicodechar{ℓ}{\ell}
\newunicodechar{•}{\item}
\newunicodechar{…}{\dots}
\newunicodechar{≔}{\coloneqq}
\newif\ifnormopen\normopenfalse
\newunicodechar{‖}{\ifnormopen\rVert\normopenfalse\else\rVert\normopentrue\fi}
\newunicodechar{≤}{\leq}
\newunicodechar{≥}{\geq}
\newunicodechar{≰}{\nleq}
\newunicodechar{≱}{\ngeq}
\newunicodechar{⊕}{\oplus}
\newunicodechar{⊗}{\otimes}
\newunicodechar{≠}{\neq}
\newunicodechar{¬}{\neg}
\newunicodechar{≡}{\equiv}
\newunicodechar{₀}{_0}
\newunicodechar{₁}{_1}
\newunicodechar{₂}{_2}
\newunicodechar{₃}{_3}
\newunicodechar{₄}{_4}
\newunicodechar{₅}{_5}
\newunicodechar{₆}{_6}
\newunicodechar{₇}{_7}
\newunicodechar{₈}{_8}
\newunicodechar{₉}{_9}
\newunicodechar{ₚ}{_p}
\newunicodechar{ₙ}{_n}
\newunicodechar{ₐ}{_a}
\newunicodechar{ₑ}{_e}
\newunicodechar{ₕ}{_h}
\newunicodechar{ₖ}{_k}
\newunicodechar{ₗ}{_l}
\newunicodechar{ₘ}{_m}
\newunicodechar{ₛ}{_s}
\newunicodechar{ₜ}{_t}
\newunicodechar{ₓ}{_x}
\newunicodechar{⁰}{^0}
\newunicodechar{¹}{^1}
\newunicodechar{²}{^2}
\newunicodechar{³}{^3}
\newunicodechar{⁴}{^4}
\newunicodechar{⁵}{^5}
\newunicodechar{⁶}{^6}
\newunicodechar{⁷}{^7}
\newunicodechar{⁸}{^8}
\newunicodechar{⁹}{^9}
\newunicodechar{ⁿ}{^n}
\newunicodechar{∈}{\in}
\newunicodechar{∉}{\notin}
\newunicodechar{⊂}{\subset}
\newunicodechar{⊃}{\supset}
\newunicodechar{⊆}{\subseteq}
\newunicodechar{⊇}{\supseteq}
\newunicodechar{⊄}{\nsubset}
\newunicodechar{⊅}{\nsupset}
\newunicodechar{⊈}{\nsubseteq}
\newunicodechar{⊉}{\nsupseteq}
\newunicodechar{∪}{\cup}
\newunicodechar{∩}{\cap}
\newunicodechar{∀}{\forall}
\newunicodechar{∃}{\exists}
\newunicodechar{∄}{\nexists}
\newunicodechar{∨}{\vee}
\newunicodechar{∧}{\wedge}
\newunicodechar{⊼}{\bar{\wedge}}
\newunicodechar{⊽}{\bar{\vee}}
\newunicodechar{⌊}{\lfloor}
\newunicodechar{⌋}{\rfloor}
\newunicodechar{⌈}{\lceil}
\newunicodechar{⌉}{\rceil}
\newunicodechar{·}{\cdot}
\newunicodechar{∘}{\circ}
\newunicodechar{×}{\times}
\newunicodechar{↑}{\uparrow}
\newunicodechar{↓}{\downarrow}
\newunicodechar{→}{\rightarrow}
\newunicodechar{←}{\leftarrow}
\newunicodechar{⇒}{\Rightarrow}
\newunicodechar{⇐}{\Leftarrow}
\newunicodechar{↔}{\leftrightarrow}
\newunicodechar{⇔}{\Leftrightarrow}
\newunicodechar{↦}{\mapsto}
\newunicodechar{∅}{\emptyset}
\newunicodechar{∞}{\infty}
\newunicodechar{≅}{\cong}
\newunicodechar{≈}{\approx}
\newunicodechar{ℓ}{\ell}
\newunicodechar{𝟙}{\mathds{1}}
\newunicodechar{𝟘}{\mathds{0}}
\newunicodechar{↪}{\hookrightarrow}

\newunicodechar{𝔸}{\mathbb{A}}
\newunicodechar{𝔹}{\mathbb{B}}
\newunicodechar{ℂ}{\mathbb{C}}
\newunicodechar{𝔻}{\mathbb{D}}
\newunicodechar{𝔼}{\mathbb{E}}
\newunicodechar{𝔽}{\mathbb{F}}
\newunicodechar{𝔾}{\mathbb{G}}
\newunicodechar{ℍ}{\mathbb{H}}
\newunicodechar{𝕀}{\mathbb{I}}
\newunicodechar{𝕁}{\mathbb{J}}
\newunicodechar{𝕂}{\mathbb{K}}
\newunicodechar{𝕃}{\mathbb{L}}
\newunicodechar{𝕄}{\mathbb{M}}
\newunicodechar{ℕ}{\mathbb{N}}
\newunicodechar{𝕆}{\mathbb{O}}
\newunicodechar{ℙ}{\mathbb{P}}
\newunicodechar{ℚ}{\mathbb{Q}}
\newunicodechar{ℝ}{\mathbb{R}}
\newunicodechar{𝕊}{\mathbb{S}}
\newunicodechar{𝕋}{\mathbb{T}}
\newunicodechar{𝕌}{\mathbb{U}}
\newunicodechar{𝕍}{\mathbb{V}}
\newunicodechar{𝕎}{\mathbb{W}}
\newunicodechar{𝕏}{\mathbb{X}}
\newunicodechar{𝕐}{\mathbb{Y}}
\newunicodechar{ℤ}{\mathbb{Z}}

\newunicodechar{α}{\alpha}
\newunicodechar{β}{\beta}
\newunicodechar{γ}{\gamma}
\newunicodechar{Γ}{\Gamma}
\newunicodechar{δ}{\delta}
\newunicodechar{Δ}{\Delta}
\newunicodechar{ε}{\varepsilon}
\newunicodechar{ζ}{\zeta}
\newunicodechar{η}{\eta}
\newunicodechar{θ}{\theta}
\newunicodechar{Θ}{\Theta}
\newunicodechar{ι}{\iota}
\newunicodechar{κ}{\kappa}
\newunicodechar{λ}{\lambda}
\newunicodechar{Λ}{\Lambda}
\newunicodechar{μ}{\mu}
\newunicodechar{ν}{\nu}
\newunicodechar{ξ}{\xi}
\newunicodechar{Ξ}{\Xi}
\newunicodechar{π}{\pi}
\newunicodechar{Π}{\Pi}
\newunicodechar{ρ}{\rho}
\newunicodechar{σ}{\sigma}
\newunicodechar{Σ}{\Sigma}
\newunicodechar{τ}{\tau}
\newunicodechar{υ}{\upsilon}
\newunicodechar{ϒ}{\Upsilon}
\newunicodechar{φ}{\varphi}
\newunicodechar{ϕ}{\phi}
\newunicodechar{Φ}{\Phi}
\newunicodechar{χ}{\chi}
\newunicodechar{ψ}{\psi}
\newunicodechar{Ψ}{\Psi}
\newunicodechar{ω}{\omega}
\newunicodechar{Ω}{\Omega}

\newunicodechar{𝒜}{\mathcal{A}}
\newunicodechar{ℬ}{\mathcal{B}}
\newunicodechar{𝒞}{\mathcal{C}}
\newunicodechar{𝒟}{\mathcal{D}}
\newunicodechar{ℰ}{\mathcal{E}}
\newunicodechar{ℱ}{\mathcal{F}}
\newunicodechar{𝒢}{\mathcal{G}}
\newunicodechar{ℋ}{\mathcal{H}}
\newunicodechar{ℐ}{\mathcal{I}}
\newunicodechar{𝒥}{\mathcal{J}}
\newunicodechar{𝒦}{\mathcal{K}}
\newunicodechar{ℒ}{\mathcal{L}}
\newunicodechar{ℳ}{\mathcal{M}}
\newunicodechar{𝒩}{\mathcal{N}}
\newunicodechar{𝒪}{\mathcal{O}}
\newunicodechar{𝒫}{\mathcal{P}}
\newunicodechar{𝒬}{\mathcal{Q}}
\newunicodechar{ℛ}{\mathcal{R}}
\newunicodechar{𝒮}{\mathcal{S}}
\newunicodechar{𝒯}{\mathcal{T}}
\newunicodechar{𝒰}{\mathcal{U}}
\newunicodechar{𝒱}{\mathcal{V}}
\newunicodechar{𝒲}{\mathcal{W}}
\newunicodechar{𝒳}{\mathcal{X}}
\newunicodechar{𝒴}{\mathcal{Y}}
\newunicodechar{𝒵}{\mathcal{Z}}
\newunicodechar{𝒶}{\mathcal{a}}
\newunicodechar{𝒷}{\mathcal{b}}
\newunicodechar{𝒸}{\mathcal{c}}
\newunicodechar{𝒹}{\mathcal{d}}
\newunicodechar{ℯ}{\mathcal{e}}
\newunicodechar{𝒻}{\mathcal{f}}
\newunicodechar{ℊ}{\mathcal{g}}
\newunicodechar{𝒽}{\mathcal{h}}
\newunicodechar{𝒾}{\mathcal{i}}
\newunicodechar{𝒿}{\mathcal{j}}
\newunicodechar{𝓀}{\mathcal{k}}
\newunicodechar{𝓁}{\mathcal{l}}
\newunicodechar{𝓂}{\mathcal{m}}
\newunicodechar{𝓃}{\mathcal{n}}
\newunicodechar{ℴ}{\mathcal{o}}
\newunicodechar{𝓅}{\mathcal{p}}
\newunicodechar{𝓆}{\mathcal{q}}
\newunicodechar{𝓇}{\mathcal{r}}
\newunicodechar{𝓈}{\mathcal{s}}
\newunicodechar{𝓉}{\mathcal{t}}
\newunicodechar{𝓊}{\mathcal{u}}
\newunicodechar{𝓋}{\mathcal{v}}
\newunicodechar{𝓌}{\mathcal{w}}
\newunicodechar{𝓍}{\mathcal{x}}
\newunicodechar{𝓎}{\mathcal{y}}
\newunicodechar{𝓏}{\mathcal{z}}
\usepackage{listings}
\usepackage{enumerate}
\usepackage{booktabs}
\usepackage{amsmath,amssymb,amsthm,mathtools}
\usepackage{bm}
\usepackage{pgfplots}
\pgfplotsset{compat=1.18} 

\newcommand{\Bin}{\mathrm{Bin}}
\newcommand{\Var}{\mathrm{Var}}

\newcommand{\bucket}{\mathrm{bucket}}
\newcommand{\slot}{\mathrm{slot}}
\newcommand{\h}{h}
\newcommand{\s}{s}
\newcommand{\seed}{\mathrm{seed}}

\usepackage[commentsnumbered,ruled,vlined]{algorithm2e} 

    \SetCommentSty{commentFont}
    \SetKwProg{algo}{Algorithm}{:}{}
    \usepackage{etoolbox}
    \patchcmd{\SetProgSty}{ArgSty}{ProgSty}{}{}
    \SetProgSty{}
    \SetKw{And}{and}
    \SetKw{Or}{or}
    \DontPrintSemicolon

\newcommand{\parag}[1]{\vspace{-1ex}\subparagraph*{#1}}

\title{Non-Minimal \texorpdfstring{$k$}{k}-Perfect Hashing: Tight Lower Bounds and an Application to Fast Static Hash Tables}
\titlerunning{Non-Minimal $\bm{k}$-Perfect Hashing}

\author{Ragnar Groot Koerkamp}{Karlsruhe Institute of Technology, Germany}{ragnar.grootkoerkamp@gmail.com}{https://orcid.org/0000-0002-2091-1237}{}
\author{Stefan Hermann}{Karlsruhe Institute of Technology, Germany}{hermann@kit.edu}{https://orcid.org/0000-0001-9183-2926}{}
\author{Peter Sanders}{Karlsruhe Institute of Technology, Germany}{sanders@kit.edu}{https://orcid.org/0000-0003-3330-9349}{}
\author{Stefan Walzer}{Karlsruhe Institute of Technology, Germany}{stefan.walzer@kit.edu}{https://orcid.org/0000-0002-6477-0106}{}
\authorrunning{R. Groot Koerkamp, S. Hermann, P. Sanders, and S. Walzer}
\Copyright{Ragnar Groot Koerkamp, Stefan Hermann, Peter Sanders, and Stefan Walzer}

\relatedversiondetails[cite={extended-version}]{Extended Version}{https://arxiv.org/abs/2607.07257}
\supplement{}
\supplementdetails[subcategory={}, cite=nonminimal-kphf-repo, swhdelimiter={\\\indent}, swhid={swh:1:dir:d8b8fa726bc8cc741a074c79b7d9bbf0ef0b69e7}]{\newline Software}{https://github.com/RagnarGrootKoerkamp/static-hash-sets}

\begin{document}

\maketitle
\begin{abstract}
A \emph{minimal perfect hash function} (minimal PHF) is a data structure mapping a static set of $n$ keys to $n$ bins without collisions.
Two natural generalizations are minimal $k$-PHFs where $n$ keys are mapped to $n/k$ bins of capacity $k$ each, and (non-minimal) PHFs with load factor $α < 1$ where the number of bins is increased by a factor of $1/α$, resulting in spare capacity.
    
While there has been a recent surge of interest in perfect hashing generally,
non-minimal $k$-PHFs have not been systematically studied despite a natural use case of speeding up static hash tables: The idea is that a small cache-resident $k$-PHF maps each key $x$ to a cache-line-sized bin of capacity $k$ where $x$ resides. Ideally, this yields a branchless lookup operation with a single cache miss working at high load factors for positive and negative queries alike.
    
Our main theoretical contribution is to determine tight space lower bounds
for $k$-PHFs for all pairs of $α \in (0,1]$ and $k \geq 1$. It turns out that
combining $α < 1$ and $k \geq 2$ drastically reduces the space of $k$-PHFs,
e.g. for $(k,α) = (16,0.8)$ the space lower bound is $0.027$ bits per key while for $(k,α) = (16,1.0)$ and $(k,α) = (1,0.8)$ the lower bounds are higher by factors of $\approx 8$ and $\approx 32$, respectively.
On the practical side, we develop a $k$-PHF based on PtrHash and tune it for use in static hash tables. 
Empirically, our implementation produces $k$-PHFs of size roughly $50\%$ above the lower bound.
A static hash set based on this $k$-PHF is consistently at least as fast as other hash
sets for negative and mixed queries. On two of the three tested
architectures it achieves
up to $1.5\times$ speedup for large $n\geq 30M$ where a $1$-PHF does not fit in
cache.
\end{abstract}

\section{Introduction}
\label{sec:org03bad15}
\begin{table}[t]
    \caption{Overview of static techniques and lower bounds. For the lower bounds we assume a large universe and omit lower order terms. For details about the techniques, we refer to the survey \cite{lehmann2025modern}.
    }
    \label{tab:overview}
    \def\arraystretch{1.2}
    \begin{tabular}{p{20mm}p{11.1cm}}
        \toprule
        $k=1$, $\alpha = 1$ &
            Bucket Placement \cite{fox1992faster, belazzougui2009hash, pibiri2021pthash, hermann2024phobic, grootkoerkamp2025ptrhash, beling2026phast}, Recursive Splitting \cite{esposito2020recsplit, lehmann2025consensus, bez2023high}, Fingerprinting \cite{beling2023fingerprinting, chapman2011meraculous, limasset2017fast, lehmann2024fast}, Hypergraph / Retrieval \cite{botelho2007external, czech1992optimal, majewski1996family, botelho2004new,genuzio2016fast, weaver2020constructing, lehmann2023sichash, lehmann2023shockhash, hermann2025morphishash}
        \vspace{1mm}\newline\textbf{Lower Bound:} $n \cdot \log₂(e)$ \cite{mehlhorn1982program}\\\hline
        $k=1 $, $\alpha < 1$ &
            Bucket Placement \cite{fox1992faster, belazzougui2009hash, pibiri2021pthash, hermann2024phobic, grootkoerkamp2025ptrhash, beling2026phast},  Hypergraph / Retrieval \cite{chazelle2004bloomier, lehmann2023sichash}
        \vspace{1mm}\newline\textbf{Lower Bound:} $n \cdot  \left( \log₂(e) + \left(\frac{1}{\alpha} - 1\right)\log₂(1 - \alpha) \right)$ \cite{belazzougui2009hash} \\\hline
        $k \geq 2 $, $\alpha = 1$ &
            Bucket Placement \& Recursive Splitting \& Thresholding \cite{hermann2025engineering}
        \vspace{1mm}\newline\textbf{Lower Bound:}     $n \cdot \left(\log₂(e) - \log₂(k^k/k!) / k\right)$ \cite{mairson1983program,belazzougui2009hash,kurpicz2023pachash} \\ \hline
        $k \geq 2 $, $\alpha < 1$
         &
            Bucket Placement \cite{belazzougui2009hash} and [\cref{sec:k-ptr}]
        \vspace{1mm}\newline\textbf{Lower Bound (new):}  $ n \cdot \left(\log₂\left(\frac{\lambda e}{\alpha k}\right) + \frac{\log₂ \zeta}{\alpha k} \right)$ [\cref{thm:precise-formula}] \newline where $\lambda$ s.t. for $𝔼_{Y \sim \mathrm{Pois}(λ)}[Y|Y ≤ k] = α k$ and $ζ:=e^{-λ}/\Pr[Y ≤ k]$.\\
    \bottomrule
    \end{tabular}
\end{table}

Given a universe of $u$ keys and an input set of $n$ keys, a $k$-perfect hash function ($k$-PHF) maps the input keys to $b$ \emph{bins} such that there are at most $k$ keys in each bin. The challenge is to represent the $k$-PHF as a data structure without storing the input set itself.
We refer to $\alpha := \frac{n}{bk}$ as the \emph{load} of a $k$-PHF and call a $k$-PHF \emph{minimal} when $\alpha=1$.
\cref{tab:overview}, as well as the survey \cite{lehmann2025modern}, give an overview of construction techniques and lower bounds.
In this paper,
we initiate the study of the most difficult case of non-minimal $k$-PHF with $k ≥ 2$.
We derive tight lower bounds on the space of a non-minimal $k$-PHF for all pairs of $α$ and $k$ (\cref{sec:lb}),
extend PtrHash \cite{grootkoerkamp2025ptrhash} into $k$-PtrHash, a fast and practical $k$-PHF with about $50\%$ space overhead (\cref{sec:k-ptr}), and
apply this $k$-PHF to obtain a fast static hash set for large inputs (\cref{s:application}).

\parag{Theory: Tight lower bounds.}
One way to construct a PHF is via a brute force algorithm:
try out seeds $s=0,1,\ldots$ of a hash function that randomly maps the keys to
the bins, until one is found that maps at most $k$ keys to each bin.
Once found, the data structure stores~$s$.
Storing $s$ takes roughly $\log_2 \frac 1q$ bits, where $q$ is the success probability of a single seed.
For the cases where $k=1$ \emph{or} $\alpha=1$, it was known that brute force is space optimal up to lower order terms \cite{belazzougui2009hash}.
The first theoretical contribution of this paper is to extend this result to all
combinations of $\alpha$ and $k$:

\begin{theorem}[informal version of \cref{thm:brute-force-is-optimal}]
    \label{thm:brute-force-is-optimal-informal}
    Assuming the universe of keys is sufficiently large, the brute force algorithm is space optimal for all $0<\alpha\leq 1$ and $k\geq 1$.
\end{theorem}

Hence, tight lower bounds can be attained by determining the probability $q$ that a random function $f : S → [b]$ from a set $S$ of $n$ keys to $b$ bins is $k$-perfect. Equivalently $q$ is the ratio of the number $P$ of $k$-perfect functions and the number $b^{-n}$ of all functions.

For $\alpha\in (0,1]$ and $k=1$ we have $P = \frac{b!}{(b-n)!}$ by standard counting arguments. A simple calculation using Stirling's approximation then yields the first two lower bounds shown in \cref{tab:overview}.
The case of $k ≥ 2$ and $α = 1$ is also easy: We have $P = \binom{n}{k\ k\ \cdots \ k} = \frac{n!}{(k!)^b}$ from which the third bound in \cref{tab:overview} follows.

The case of $k ≥ 2$ and $α < 1$ is more difficult. The issue is that non-minimal $k$-perfect functions are more varied with an unknown number of bins of the various loads in $\{0,…,k\}$. While we can still write a formula like
\[
    P = \sum_{\substack{x₁,…,x_b ∈ \{0,…,k\}\\x₁+…+x_b = n}} \binom{n}{x₁\,x₂\,…\,x_b},
\]
such a formula is not useful, for instance it does not reveal the asymptotic behaviour of $q$.
Previous work has also stated that the case of $k\geq 2$ and $\alpha \in (0,1)$ appears non-trivial \cite{belazzougui2009hash}.
Using a new approach, our main theoretical contribution is to determine the success probability $q$ for any combination of $k$ and $\alpha$. This gives:

\begin{theorem}[informal version of \cref{thm:precise-formula}]
  For any $α ∈ (0,1)$ and $k ≥ 1$, the tight space lower bound for non-minimal $k$-PHF is
\[
    n \cdot \left(\log₂\left(\frac{\lambda e}{\alpha k}\right) + \frac{\log₂ \zeta}{\alpha k} \right)
\]
where $\lambda = λ(α,k)$ and $ζ = ζ(α,k)$ are as defined in \cref{tab:overview}.
\end{theorem}

The theorem implies that the optimal brute-force algorithm
results in a situation where the number of keys in each bin follows a \emph{truncated Poisson distribution} with parameter $λ$ (taking only values $\leq k$) and expectation $αk$. (The distribution is slightly skewed by the global condition that the bin loads sum to $n$.)

We also derive a second result on the space requirement of $k$-PHFs that neglects constant factors but has the advantage of yielding closed form expressions. The result involves the \emph{slack} $s = ⌊(1-α)k+1⌋$, which is intuitively the average number of free slots per bin. The result reveals two regimes, depending on whether the slack is higher or lower than the standard deviation $\sqrt{k}$ of bin loads in random assignments.

\begin{theorem}[Informal version of \cref{t:asymp}]
    Let $S(α,k)$ be the tight space lower bound for a $k$-PHF in bits per bin. There are the following cases:\\
    \begin{tabular}{l@{\ }l}
        If $s\ll \sqrt k$ & then $S(α,k) = \Theta(\log(\sqrt k /s))$.\\
        If $s=\Theta(\sqrt k)$ & then $S(α,k) = \Theta(1)$.\\
        If $s\gg \sqrt k$ & then $S(α,k) = \exp(-\Theta(s^2/k))$.\\
    \end{tabular}
\end{theorem}

\parag{Practice: $k$-PHF using $k$-PtrHash.}
The three relevant performance metrics of a PHF are construction throughput, space usage, and query throughput.
A variety of algorithms exists for the case $k=1$, covering a large spectrum of trade-offs between the three metrics.
While some techniques are extremely close to the lower bound \cite{lehmann2025consensus} ($\alpha=1$), others offer fast construction and queries \cite{grootkoerkamp2025ptrhash, beling2026phast}.
We refer to the recent survey on minimal $1$-PHFs for details \cite{lehmann2025modern}.
A number of methods was recently generalized to minimal $k$-PHFs for $k\geq 2$ \cite{hermann2025engineering},
while the case $(k\geq 2, \alpha<1)$ has only received minimal attention \cite{belazzougui2009hash}.
The only technique that has so far been applied to all combinations of $\alpha$ and $k$ is \emph{bucket placement}.
In this paper, we implement $k$-PtrHash (\cref{sec:k-ptr}), a natural generalization of the PtrHash bucket placement technique \cite{grootkoerkamp2025ptrhash} that is tuned for high query throughput and space efficiency.

Bucket placement works by first partitioning the key set into small buckets.
Buckets are then processed one-by-one:
For each bucket, a \emph{seed} (of a hash function) is found where all keys of that bucket are hashed (placed) to the bins such that there are at most $k$ keys in each bin.
For $k=1$, this allows a natural greedy algorithm that chooses for each bucket the first working seed, and this method was also used for $k\geq 2$ in \cite{hermann2025engineering} and \cite{belazzougui2009hash}.
For $k\geq 2$, however, this is suboptimal and makes placing future keys harder than
needed:
whereas for $k=1$ the only ``state'' is the number of full bins, for $k\geq 2$,
there is a ``quality'' to the distribution of the keys so far.
For example, consider $k=2$ and $\alpha=1$.
At the halfway point, we might have filled exactly one of the two slots in each bin, in which case there is still a lot of freedom to map the remaining keys.
If, on the other hand, exactly half of the bins are filled while the other half is empty, mapping the remaining keys is much more restricted.
Thus, one should prefer seeds that result in a more even ``spread'' of the keys over the bins, and we introduce a \emph{scoring function} to express this preference.

Our implementation $k$-PtrHash reaches as low as $50\%$ space overhead over the lower bound, while reaching a throughput of over 300 million queries per second.

\parag{Application: a static hash table based on $k$-PHF.}
A natural application of $k$-PHFs are static hash sets%
\footnote{In fact, the practical part of this paper was motivated by Deacon
\cite{deacon}, which is a tool that stores a static subset of $300$ million of
the $k$-mers in a human genome and then queries this set to determine whether pieces of DNA are human or not.}
(and tables):
Given a set of keys, one can first construct a $k$-PHF and then place the keys into their bins.
A natural choice for $k$ is the number of keys that fit into a cache line.
Given a query key, we evaluate the $k$-PHF, retrieve the single corresponding cache line and then answer the query using SIMD.
A non-minimal $8$-PHF with $\alpha = 0.9$ requires $7\times$ times less space than a $1$-PHF  with $\alpha = 0.9$.
The reduced space usage allows the $k$-PHF to fit in the cache for larger
inputs.
For large inputs, the $k$-PHF-set is consistently the fastest for negative and
mixed queries and achieves up to $1.5\times$ higher query throughput on two of
the three tested architectures (\cref{s:application}).

\section{Tight Lower Bounds for non-minimal \texorpdfstring{$\bm k$}{k}-PHFs}
\label{sec:lb}
\begin{figure}[t]
    \centering
    \input{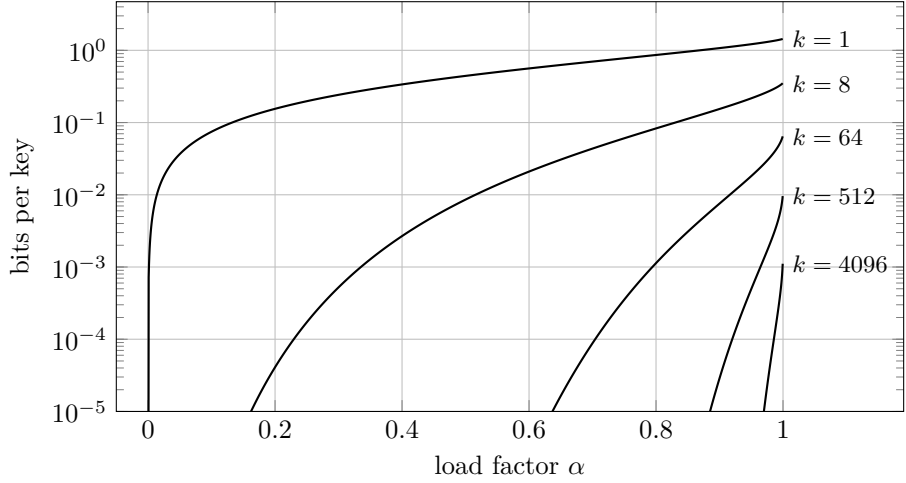}
    \caption{Space lower bound for different load factors and $k$.}
    \label{fig:space_plot}
\end{figure}

In this section we determine tight lower bounds on the space requirement of non-minimal \texorpdfstring{$k$}{k}-PHFs. In \cref{sec:brute-force-optimal} we show that a simple brute force algorithm is space-optimal, which is neither surprising nor particularly difficult. What \emph{is} difficult is computing the expected space usage of that brute force algorithm, which we do in \cref{sec:brute-force-space-usage}. Our resulting formula is exact (up to lower order terms) but unwieldy. We therefore also derive explicit formulas that neglect constant factors in \cref{sec:space-usage-explicit-formulas}. These use completely different arguments.

\begin{table}[t]
\centering
\caption{Space lower bounds for different $k$ and $\alpha$ in bits per key,
  as computed by the script in \cref{app:sagemath-script}.}
\label{tab:space_values}
\begin{tabular}{l@{\hspace{1cm}}rrrrrrr}
  \toprule
  $\alpha$ & $k=1$ & $k=2$ & $k=4$ & $k=8$ & $k=16$ & $k=32$ & $k=64$ \\
  \midrule
  $0.20$ & 0.154983 & 0.031120 & 0.002591 & 0.000041 & $\approx 0$ & $\approx 0$ & $\approx 0$ \\
  $0.40$ & 0.337247 & 0.112213 & 0.024445 & 0.002680 & 0.000082 & $\approx 0$ & $\approx 0$ \\
  $0.60$ & 0.561410 & 0.244721 & 0.084103 & 0.020956 & 0.003182 & 0.000199 & 0.000002 \\
  $0.80$ & 0.862213 & 0.456247 & 0.210448 & 0.082945 & 0.026996 & 0.006774 & 0.001132 \\
  $0.90$ & 1.073592 & 0.621795 & 0.326312 & 0.154184 & 0.064989 & 0.024040 & 0.007553 \\
  $0.95$ & 1.215225 & 0.739728 & 0.416518 & 0.216630 & 0.103742 & 0.045522 & 0.018156 \\
  $0.99$ & 1.375585 & 0.880564 & 0.533189 & 0.306808 & 0.168119 & 0.087660 & 0.043398 \\
  $1.00$ & 1.442695 & 0.942695 & 0.588936 & 0.355096 & 0.208329 & 0.119672 & 0.067619 \\
  \bottomrule
\end{tabular}
\end{table}

\subsection{Brute Force is Space-Optimal}
\label{sec:brute-force-optimal}

\def\comp{\mathrm{comp}}
Once constructed, a $k$-PHF behaves like a function $f : [u] → [b]$ where $f(x)$ is the bin returned for the key $x$. The $k$-PHF must be \emph{compatible} with its input set $S ⊆ [u]$, meaning $|f^{-1}(i) ∩ S| ≤ k$ for all $i ∈ [b]$. Let $\comp(f)$ be the number of input sets of size $n$ that are compatible with $f$. Assume for simplicity that $b$ divides $u$ and call $f$ \emph{balanced} if $|f^{-1}(1)| = … = |f^{-1}(b)|$.

\pagebreak
Consider the following intuitive statement.
\begin{restatable}{lemma}{balancedOptimal}
    \label{lem:balanced-is-optimal}
    $\comp(f)$ is maximised if $f$ is balanced.
\end{restatable}

\def\qsimp{q_{\mathrm{simp}}}
\def\qbal{q_{\mathrm{bal}}}

By standard arguments it suffices to prove the following simpler lemma.
\begin{lemma}
    Consider a set of $u$ distinguishable balls of various colours, including $R$ red balls and $B$ blue balls with $R ≥ B+2$. Consider the number of subsets of size $n$ that involve at most $k$ balls of each colour. This number is non-decreasing if we recolour a red ball blue.
\end{lemma}
\begin{proof}
    It suffices if we prove the lemma assuming that red and blue are the only available colours. This is because for whatever subset $S$ of balls of the other colours we consider, the number of ways to extend $S$ to a valid subset overall is non-decreasing if we do the recolouring.
    
    Hence assume $u = R+B$ and fix a special red ball that is to be recoloured. We may also assume $B ≥ k$ and $k < n ≤ 2k$ as otherwise the claim is trivial. The subsets of $n$ balls fall into the following categories.
    \begin{itemize}
        \item neutral subsets where the special ball is not drawn,
        \item neutral subsets where the colour of the special ball does not make a difference,
        \item $\binom{R-1}{k}\binom{B}{n-k-1}$ good subsets where the special ball would have been a red ball to many,
        \item $\binom{R-1}{n-k-1}\binom{B}{k}$ bad subsets where the special ball will be a blue too many.
    \end{itemize}
    We compute that the good cases outnumber the bad cases.
    \begin{align*}
        \frac{\text{good}}{\text{bad}} &= \frac{\binom{R-1}{k}\binom{B}{n-k-1}}{\binom{R-1}{n-k-1}\binom{B}{k}} = \frac{(R - 1 - n + k + 1)! (B - k)!}{(R - 1 - k)! (B - n + k +1)!}\\
        &=\frac{(B-k) · \ldots · (B-n+k+2)}{(R-1-k)· … · (R-1-n+k+2)}
        \geq 1\cdot \ldots \cdot1 = 1 \qedhere
    \end{align*}
\end{proof}

Let now $\qbal = \qbal(u,n,b,k)$ be the probability that a uniformly random input set is compatible with a fixed balanced function, or equivalently%
\footnote{To see the equivalence, first define $\qbal$ as the probability that a random input set is compatible with a random balanced function. Then note the symmetry between all input sets and the symmetry between all balanced functions. These symmetries imply that we can condition on any input set or any balanced function without affecting $\qbal$.},
the probability that a fixed input set is compatible with a uniformly random balanced function.

\begin{theorem}
    \label{thm:balanced-brute-force-is-optimal}
    The \emph{balanced brute force algorithm} (see below) constructs a $k$-PHF with expected space $\log₂ \frac 1\qbal + 𝒪(1)$, with $\qbal$ as defined above, and no (randomised) construction can undercut this bound on expected space by more than $𝒪(1)$ bits.
\end{theorem}

\begin{proof}
    The balanced brute force algorithm tries a sequence $f₁,f₂,…$ of uniformly random balanced functions and stores the smallest index $S$ for which $f_S$ is compatible with the input set. This seed has distribution $S \sim \mathrm{Geom}(\qbal)$ and can be stored in expected space $\log₂ \frac 1\qbal + 𝒪(1)$ using Rice encoding.
    
    Now we show that no other construction can do significantly better. Consider any deterministic construction $𝒟$ first. Let $ℱ ⊆ [b]^{[u]}$ be the family of functions that $𝒟$ may produce over all possible input sets combined. Any $f ∈ ℱ$ satisfies $\comp(f) ≤ \qbal·\binom{u}{n}$ by \cref{lem:balanced-is-optimal} and the definition of $\qbal$. Moreover, $\sum_{f ∈ ℱ} \comp(f) ≥ \binom{u}{n}$ as every input set must be compatible with some $f ∈ ℱ$. This implies $|ℱ| ≥ \frac 1\qbal$. To represent every possible $f ∈ ℱ$, at least $\log₂ |ℱ| ≥ \log₂ \frac{1}{\qbal}$ bits are needed in the worst case.
    
    It is not hard to see that if the input is uniformly random (rather than worst case), the expected space requirement drops by at most $𝒪(1)$ bits.\footnote{This uses that if $Y \sim \mathrm{Unif}(\{1,…,N\})$ then $𝔼[\log₂ Y] = \log₂ N - 𝒪(1)$.} By Yao's Theorem, $\log₂ \frac 1\qbal - 𝒪(1)$ is therefore also a lower bound for the expected space requirement of any randomised algorithm.
\end{proof}

Now consider the related \emph{simple brute force} algorithm that tries uniformly random functions that are not necessarily balanced. Under a weak assumption simple brute force is also optimal.

\begin{restatable}{theorem}{simpleOptimal}
    \label{thm:brute-force-is-optimal}
    Assume $u ≥ n²/α$. Then simple brute force is space optimal in the same sense as balanced brute force (see \cref{thm:balanced-brute-force-is-optimal}).
\end{restatable}
The proof is standard and found in \cref{app:simple-brute-force}. To see that a lower bound on $u$ is required, consider the silly edge case of $u = n$ and $α = k = 1$. The identity function is minimal perfect and requires no space to store, while simple brute force requires $\log₂\frac1{n!} ≈ 1.44n$ bits.

\subsection{Precise Formulas and the Success Probability of Simple Brute Force}
\label{sec:brute-force-space-usage}

\def\TT{\widetilde{Θ}}

In this section we approximate the probability $\qsimp$ that a uniformly random function is $k$-perfect on a given input set. This leads to the following result.
\begin{theorem}
    \label{thm:precise-formula}
    Let $k,n,u ∈ ℕ$ and $α ∈ (0,1)$ with $u ≥ n²/α$.
    Let $λ$ be such that $𝔼_{Y \sim \mathrm{Pois}(λ)}[Y|Y ≤ k] = α k$ and $ζ:=e^{-λ}/\Pr[Y ≤ k]$.
    Then the best possible space requirement for $k$-PHFs with load factor $α$ is
    \[
        \underbrace{\log₂ \frac{λe}{αk} + \frac{\log₂ζ}{αk}}_{S(α,k) :=} \pm \widetilde{𝒪}(\tfrac 1n) \text { bits per key where $\widetilde{𝒪}$ ignores terms only depending on $α$ and $k$.}
    \]
\end{theorem}
Note that $S(α,k)$ is independent of $n$ and $u$. What makes the formula a bit unwieldy is that $ζ$ and $λ$ are non-explicit functions of $α$ and $k$. Numerical methods (see \cref{app:sagemath-script}) nevertheless allow us to derive the values in \cref{fig:space_plot} and \cref{tab:space_values}.

The main ingredient for proving \cref{thm:precise-formula} is the following lemma.
\begin{lemma}
    \label{l:pe}
    In the setting of \cref{thm:precise-formula} we have
    $\displaystyle
    \qsimp=\left(\frac{\alpha k}{e\lambda}\right)^n ζ^{-b} \cdot \TT(1)
    $.
\end{lemma}
\begin{proof}[Proof of \cref{thm:precise-formula}]
    By \cref{thm:balanced-brute-force-is-optimal}, the best possible space requirement for $k$-PHFs is $\log₂ \frac 1{\qbal} \pm 𝒪(1)$ and by \cref{thm:brute-force-is-optimal} we have $\qsimp = Θ(\qbal)$ when $u ≥ n²/α$ as we assume. We obtain the desired result from the claim of \cref{l:pe} by taking $\log₂ \frac{1}{·}$ on both sides, using $b = \frac{n}{αk}$, and dividing by $n$.
\end{proof}


\begin{figure}[t]
\begin{minipage}{0.25\textwidth}
    \centering
    \includegraphics[width=\textwidth]{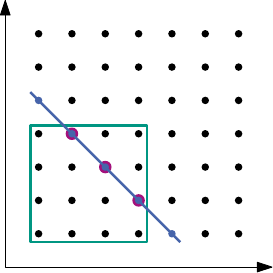}
\end{minipage}
\hfill
\begin{minipage}{0.65\textwidth}
    \caption{The two dimensional illustration uses $b=2$, $k = 4$ and $n = 5$. The multinomial $(X_1, …, X_b)$ is distributed on the blue ``diagonal'', since $X_1 + … + X_b = n$ is fixed. We wish to compute the probability that $(X₁,…,X_b)$ falls inside the green box (is “$k$-perfect”). To help with this, we consider a sequence $(Z_1, …, Z_b)$ of truncated Poisson random variables that automatically fall within the box, but not on the diagonal.}
    \label{fig:proofsketch}
\end{minipage}
\end{figure}


\begin{proof}[Proof of \cref{l:pe}]
    Consider a uniformly random function $f : S → [b]$ where $S$ is an input set of size $n$. For $i ∈ [b]$, let $X_i := f^{-1}(i)$ be the number of keys assigned to bin $i$. Clearly $\vec{X} := (X₁,…,X_b)$ follows a multinomial distribution (the bin loads after ``throwing $n$ balls into $b$ bins''). Let $R := \{(x₁,…,x_b) ∈ \{0,…,k\} \mid x₁ + … + x_b = n\}$ be the set of possible outcomes for $\vec{X}$ that correspond to $f$ being $k$-perfect on $S$. In particular we have $\qsimp = \Pr[\vec{X} ∈ R]$. For any particular $\vec{x} = (x₁,…,xₙ) ∈ R$ we have
    \[
        \Pr[\vec{X} = \vec{x}] = \binom{n}{x₁\,…\,x_b}b^{-n} = \frac{n!}{x₁! · … · x_b!} b^{-n}
    \]
    and by summing over all $x ∈ R$ we obtain
    \begin{align}
        \label{eq:pre}
        \qsimp = \Pr[\vec{X} \in R]
        &= \sum_{\vec{x} \in R} \Pr[\vec{X} = \vec{x}]
        =  \frac{n!}{b^n} \sum_{\vec{x} \in R} \prod_{i=1}^b \frac{1}{x_i!}.
    \end{align}
    We will now consider a second random sequence $\vec{Z} = (Z₁,…,Z_b)$ such that $\vec{Z}$ conditioned on $R$ has the same distribution as $\vec{X}$ conditioned on $R$. See \cref{fig:proofsketch} for an illustration.
    Concretely let $Z_1, \ldots, Z_b$ be independent with probability mass function
    \begin{equation}\label{eq:binom}
        \Pr[Z = i] =
        \begin{cases}
            0 & i > k \\
            ζ \frac{\lambda^i}{i!}  & i \leq k.
        \end{cases}
    \end{equation}
    where $λ$ is chosen such that $𝔼[Z_i] = αk$ and $ζ$ is a normalisation factor. Let
    $\vec{Z} = (Z_1, \ldots, Z_b)$ and $N_Z = \sum_{i=1}^b Z_i$.
    By construction, the events $\vec{Z} \in R$ and $N_Z = n$ are
    equivalent. For any $\vec{x} \in R$ we can compute

    \begin{align*}
        \Pr[\vec{Z} = \vec{x} &\mid N_Z = n] = \frac{\Pr[\vec{Z} = \vec{x} \wedge N_Z = n]}{\Pr[N_Z = n]} = \frac{\Pr[\vec{Z} = \vec{x}]}{\Pr[N_Z = n]} = \frac{\prod_{i=1}^b \Pr[Z_i = x_i]}{\Pr[N_Z = n]} \\
        &= \frac{\prod_{i=1}^b ζ \frac{\lambda^{x_i}}{x_i!}}{\Pr[N_Z = n]} = \frac{ζ^b\lambda^n}{\Pr[N_Z = n]} \prod_{i=1}^b \frac{1}{x_i!}.
    \end{align*}
    By summing this equation over all $\vec{x} \in R$ we get

    \begin{equation}
        \label{eq:sum-of-truncated-poisson-probs}
        1 = \frac{ζ^b\lambda^n}{\Pr[N_Z = n]}\sum_{\vec{x} \in R} \prod_{i=1}^b \frac{1}{x_i!}
    \end{equation}
    As $𝔼[N_Z] = b·𝔼[Z₁] = bαk = n$ the probability of the likely event $N_Z = n$ seems insignificant compared to the exponential terms. Indeed, using a local limit theorem for discrete random variables \cite[Theorem 1.4]{classical-local-limit-theorems} (originally in \cite{local-limit-theorem}) gives $\Pr[N_Z = n] = Θ(1/(σ\sqrt{b}))$ where $σ² = \Var(Z₁)$. We may hence use $\Pr[N_Z = n] = \TT(1/\sqrt{n})$.
    We now rearrange \cref{eq:sum-of-truncated-poisson-probs} for $\sum_{\vec{x} \in R} \prod_{i=1}^b \frac{1}{x_i!}$ and plug the result into \cref{eq:pre}.
    Using Stirling's approximation we have
    \begin{align*}
        \qsimp &= \frac{n!}{b^n} \cdot \frac{\Pr[N_Z = n]}{ζ^b\lambda^n} = \frac{n!}{b^n \lambda^n} ζ^{-b} \cdot  \TT(1/\sqrt{n}) \\
        &= \left(\frac{n}{eb\lambda}\right)^n \cdot ζ^{-b} \cdot \TT(1)= \left(\frac{\alpha k}{e\lambda}\right)^n \cdot ζ^{-b} \cdot \TT(1).\qedhere
    \end{align*}
\end{proof}

    


\subsection{Tight Space Lower Bounds up to Constant Factors}
\label{sec:space-usage-explicit-formulas}

We now derive simple \emph{explicit} formulas for the space requirement of $k$-PHFs that neglect constant factors. In the following we use $s := ⌊(1-α)k+1⌋ = |ℕ ∩ [αk,…,k]|$ to refer to the \emph{slack}, which is similar to the number $(1-α)k$ of spare slots per bin, though the definition ensures $s ≥ 1$ even for $α = 1$.
\begin{restatable}{theorem}{approxThm}
    \label{t:asymp}
    Let $k ∈ ℕ$,
    $α ∈ [\frac 12,1]$, 
    $s := ⌊(1-α)k+1⌋$ and $n ≥ n₀(k,α)$ sufficiently large. Assume $u ≥ n²/α$. The best possible space requirement for $k$-PHFs with load factor $α$ is
    \begin{align*}
        &\text{for $s = Θ(\sqrt{k})$:} &\qquad Θ(1) &\text{ bits per bin} & &\text{(total of $Θ(b) = Θ(\tfrac{n}{αk})$ bits),}\\
        &\text{for $s \ll \sqrt{k}$:} &\qquad Θ(\log(\sqrt{k}/s)) &\text{ bits per bin} & &\text{(total of $Θ(b\log(\sqrt{k}/s))$ bits),}\\
        &\text{for $s \gg \sqrt{k}$:} &\qquad \exp(-Θ(s²/k)) &\text{ bits per bin} & &\text{(total of $Θ(b\exp(-Θ(s²/k)))$ bits).}
    \end{align*}
\end{restatable}
Recall that the space requirement of simple brute force is optimal at $\log₂ \frac{1}{\qsimp} \pm Θ(1)$ bits by \cref{thm:brute-force-is-optimal} where $\qsimp$ is the probability that a fully random assignment of $n$ keys to $b$ bins results in loads $(X₁,…,X_b)$ satisfying $\max_{i ∈ [b]} X_i ≤ k$. Note that, while they are not independent, individually the loads $X_i$ have distribution $\Bin(n,\frac 1b)$.

The intuition for the bounds in \cref{t:asymp} is as follows. In the \emph{heavily loaded case} ($s \ll \sqrt{k}$) the extra number $s$ of slots per bin is much less than the standard deviation $\sqrt{k}$ of each $X_i$. The probability that a bin is neither overflowing nor significantly underloaded is $\Pr[X_i ∈ [k-2s,k]] = Θ(s/\sqrt{k})$. This suggests $\qsimp ≈ Θ((s/\sqrt{k})^b)$ leading to a space of $\log₂ \frac 1\qsimp = Θ(b\log(\sqrt{k}/s))$ as claimed.

In the \emph{lightly loaded case} ($s \gg \sqrt{k}$) the probability that an individual bin overflows is $\exp(-Θ(s²/k))$ by Chernoff bounds and minimal intervention at the overflowing bins requires $Θ(1)$ bits to fix each of them. The lower bound will exploit the following lemma, which happens to not be wasteful in this case.
\begin{lemma}
    \label{lem:negative-association}
    $\Pr[\max X_i ≤ k] ≤ \Pr[X₁ ≤ k]^b$.
\end{lemma}
\begin{proof}
    We exploit that $X₁,…,X_b$ are what is known as \emph{negatively associated} \cite[Section 3.1]{joag1983negative}. This implies that they are \emph{negative lower orthant dependent} \cite[Section 2, Property 3]{joag1983negative}, which is the name of the property we have stated.
\end{proof}
Due to space constraints the lengthy proof of \cref{t:asymp} was moved to \cref{app:approx-bounds}.

\pagebreak
\section{Extending PtrHash to \texorpdfstring{$\bm{k}$}{k}-perfect hashing}\label{sec:k-ptr}
We now extend PtrHash \cite{grootkoerkamp2025ptrhash} to non-minimal $k$-perfect hashing.

\subsection{The original PtrHash}
We briefly recall how PtrHash works.
Like its precursors PTHash \cite{pibiri2021pthash} (which improved FCH
\cite{fox1992faster} and CHD \cite{belazzougui2009hash}) and PHOBIC
\cite{hermann2024phobic}, PtrHash is based on \emph{bucket placement}.
We are given \(n\) keys \(\{k_1, \dots, k_n\}\), and let \(\alpha\leq 1\) be the target
load factor, so that the 1-PHF injectively maps the \(n\) keys to \(b=n/\alpha\) slots.
In PtrHash, values between $n$ and $n/\alpha$ are then \emph{remapped} into \([n]\), but this is not needed for \emph{non-minimal} PHFs.
The keys are hashed using a 64-bit hash \(\h_1\) and non-uniformly mapped to \(B=\lceil n/\lambda\rceil\) \emph{buckets} of
average size \(\lambda\) (in practice around $3$) using a \emph{bucket
assignment function} \(\gamma : [0, 1) \mapsto [0, 1)\) such as \(\gamma(x) = x^2\).
The data structure allocates an 8-bit \emph{seed} \(\s_j\) for
each bucket $j$ that controls the
slot that the keys finally hashes to using a second hash function \(\h_2\). To summarize in formulas:
\begin{align}
\bucket(k_i) &:= \lfloor B\cdot \gamma(\h_1(k_i)/2^{64})\rfloor,\notag\\
\seed(k_i) &:= \s_{\bucket(k_i)},\\
\slot(k_i) &:= \h_2(k_i, \seed(k_i)) \bmod b\notag.
\end{align}

\parag{Construction.}
The construction algorithm first distributes the keys over the buckets, and then
processes the buckets from large to small. For each bucket, it tries seeds from
\(0\) to \(2^8-1\) until one is found such that the corresponding slots are all still free.
The first working seed is greedily chosen and the corresponding slots are marked as taken.

\parag{Hash-evict.}
If none of the seeds work, the seed with the minimal number of colliding slots is
found and fixed. Colliding slots are handled by \emph{evicting} the keys that previously
mapped there: the seed for the colliding buckets is unset and these buckets are
pushed on a queue. For each bucket in the queue, a new search for a seed is
done, until the queue becomes empty.
Conceptually, this eviction strategy is similar to cuckoo-hashing
\cite{pagh2004cuckoo} as it spreads the complexity of hard buckets over multiple seeds.
This also ensures that the seed bits for easy-to-place buckets are also used efficiently.
Furthermore, queries remain as simple as a lookup of a single seed.

\enlargethispage{2em}
\subsection{Greedy \texorpdfstring{$\bm k$}{k}-PHF}

We naturally extend PtrHash to the case of $k$-perfect hashing by replacing
\(b=\lceil n/\alpha\rceil\) \emph{slots} by \(b=\lceil n/\alpha/k\rceil\)
\emph{bins} consisting of \(k\) slots each, similar to \cite{belazzougui2009hash}. Whereas before it was sufficient to track which
slots are already taken, now we count for each bin how many of its slots are
taken. For each bucket, we greedily choose the first seed for which none of the
mapped-to bins overflows, similar to the bucket placement method described in \cite{hermann2025engineering}.

\parag{Choosing parameters.}
Given \(\alpha\) and \(k\), we compute the space lower bound \(\ell\) in bits per key using \cref{thm:precise-formula}.
Then, we choose the target space usage as e.g. \(2\ell\), so
that each bucket of size \(8\) bits should contain on average \(\lambda =
8/(2 \ell)\) keys.

\parag{Non-minimal PHF.}
A difference compared to PtrHash is that in our setting, we do not require a
\emph{minimal} PHF.
In particular, for PtrHash, fixing (remapping) keys mapping to \([n,
n/\alpha)\) is relatively expensive in terms of space, while now we can afford to
simply \emph{not} fix them at the cost of a slightly lower load factor.
Thus, aside from reducing the space lower bound,
non-minimal PHFs allow for more freedom in the design of the data structure.

\parag{Bumping.}
Problematic buckets where none of the \(2^8\) seeds work are handled differently compared to PtrHash.
Specifically, the eviction strategy used in PtrHash has the drawback that it needs to
store a lot of metadata to know for each slot the conflicting key, and then also a lot
of time is spent repeatedly evicting already placed buckets.
Here, we take the much simpler approach of \emph{bumping} problematic buckets,
as also done by e.g. PHast \cite{beling2026phast}.
We reserve the seed value \(0\) for this case, and let it indicate that all keys
in the bucket
are \emph{bumped}, which means that they will be assigned to a bin using a
secondary/fallback data structure. We map the \(n'\) bumped keys recursively using the same
algorithm, but with lenient parameters of load factor \(\alpha' = 0.5\) and double
the target number of bits per key.
Note that the recursive structure allocates both additional memory for its
seeds, and also increases the number of target bins, thereby decreasing the
effective load factor \(\alpha\).

At query time, bumped keys are looked up in the fallback structure, and the
returned index is added to the number of bins \(b\) in the main structure.

\parag{Bucket function.}
With $k$-perfect hashing, the first keys are easier to place than the
later ones.
To achieve equal difficulty in each bucket, the keys are distributed non-uniformly.
Hermann et al. \cite{hermann2025engineering} introduce the optimal bucket function
for \emph{minimal} $k$-PHFs, which is complex to evaluate numerically.
Here, we heuristically take \(\gamma(x) = x^4\) as a fast-to-compute function
that works well despite its simplicity.

An additional benefit of such a skew function is that, for example, half the keys are mapped
to the first $6.25\%$ of buckets, so that a tiny fraction of the memory can
answer a large fraction of queries, reducing the number of cache misses.

\subsection{\texorpdfstring{$\bm k$}{k}-PtrHash: Scoring \texorpdfstring{$\bm k$}{k}-PHF}

A drawback of the greedy approach is that it uses the first seed that works,
even though there might be multiple seeds that work and some of those other candidates might be ``better''.
Specifically, we want to avoid completely filling up bins to $100\%$ load as long as possible:
this gives future keys more flexibility and therefore increases the chances that
future buckets find some working seed.
Thus, a seed is better if it places the keys into relatively empty bins, and
our non-greedy variant chooses the \emph{best} seed out of all \(2^8-1\)
candidates using a scoring heuristic.

\parag{Scoring function.}
Specifically, we introduce a \emph{scoring function} \(f: \{0,\dots, k-1\} \to
\mathbb R\) that assigns a score to adding a key to a bin that currently contains
\(i\) keys. The score of a seed is the sum of these scores over the target bins of
all keys in the bucket, and we choose the seed minimizing this score.
When multiple keys map to the same bin, we add \(f(i) + f(i+1) + \dots\) rather than multiple times \(f(i)\).
We heuristically choose the scoring function \(f(i):= 2^i\),
which means that
putting keys in nearly-full bins is heavily penalized compared to using
near-empty bins.
Experimentation shows that a wide range of functions performs
well, and e.g. \(1.5^i\) or \(6^i\) perform almost as good for the parameters
we consider ($4\leq k\leq 16$, $0.7\leq \alpha \leq 0.99$).
We choose \(f(i)=2^i\) for its simplicity.

\parag{Engineering the construction algorithm.}
Like PHast, we use a bitvector to indicate full bins and choose $h_2$ such that
the seeds in $[1, 128)$ and $[128, 256)$ respectively map each key to consecutive bins.
This allows us to use bitmasks to quickly filter out seeds that lead to collisions.
Further optimizations are explained in \cref{app:engineering}.

\subsection{Results}

$k$-PtrHash is available at
\href{https://github.com/RagnarGrootKoerkamp/static-hash-sets}{\texttt{github.com/RagnarGrootKoerkamp/static-hash-sets}}
and written in Rust.
\Cref{kphf-tuning} shows results of constructing $k$-PtrHash for various \(k\in\{4,8,16\}\), \(\alpha \in \{0.7,0.8,0.9,0.99\}\),
and a target space usage of \(1.0\times\) to \(2.5\times\) the lower bound (the
actual space usage is higher if bumping occurs). \Cref{kphf-tuning-table} shows
numeric results for the construction with scoring function, corresponding to the
data points marked in red that target $1.5\times$ the space lower bound.
We see that the scoring variant bumps less than the greedy variant, and can
achieve as low as \(1.5\times\) space overhead (red dots) with minimal (\(<0.1\%\))
bumping for \(\alpha \leq 0.9\).
This is similar to the compact configuration of PtrHash, which uses \(2\) bits per
key for \(k=1\) and \(\alpha=0.99\), which is \(1.46\times\) the lower bound of
\(1.37\) bits per key.
We further see that the greedy algorithm is faster
during construction as it has fewer seeds to consider.
Maybe counterintuitively, the scoring version gets
\emph{slower} as the space usage goes up and the problem becomes easier. This is
because there are more collision-free candidate seeds whose score has to be computed.
As expected, queries are faster when less space is used, but slow down once more
than \(\approx 1\%\) of the keys are bumped.

\begin{figure}[t]
\centering
\makebox[\linewidth]{
\includegraphics[width=1.1\linewidth]{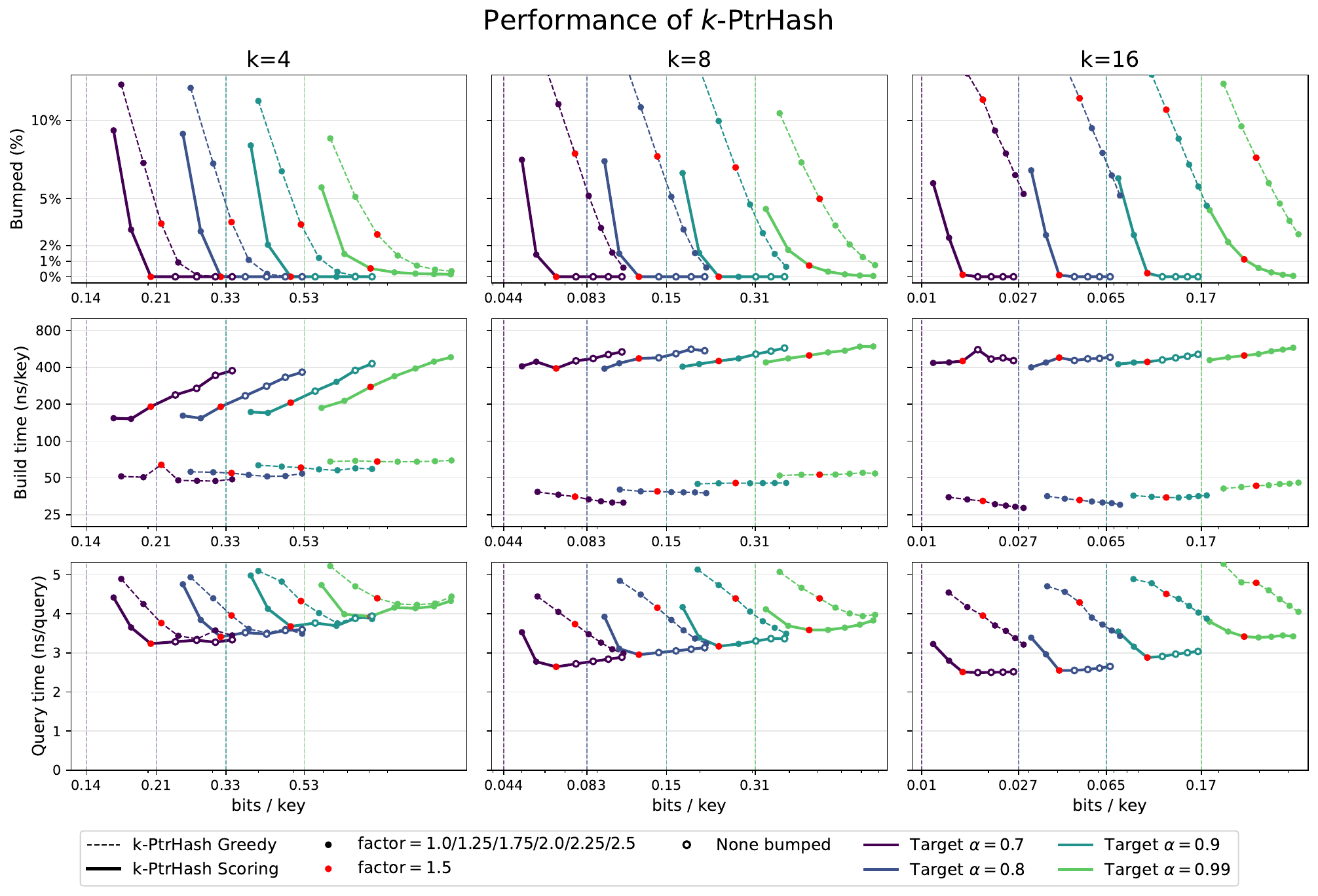}
}
\caption{\label{kphf-tuning}Results of constructing (greedy) $k$-PtrHash for \(n=10^8\) keys for \(k\in\{4,8,16\}\). Colours indicate the target load factor \(\alpha\in\{0.7,0.8,0.9,0.99\}\). Different points along each line indicate the target space usage relative to the lower bound, with the \(1.5\times\) overhead points highlighted in red. The logarithmic x-axis shows the final space usage in bits per key. The top row shows the percentage of bumped keys, with open circles indicating no bumping at all. The scoring variant bumps much less than the greedy variant. The second row shows the (logarithmic) construction time, which is \(4\times\) to \(10\times\) faster for the greedy variant. Lastly, the bottom row shows the query time, which is faster for smaller structures, but grows again when keys are bumped.}
\end{figure}

\begin{table}[t]
	\centering
	\caption{\label{kphf-tuning-table}Results of constructing the scoring
      variant of $k$-PtrHash on $n=10^8$ keys, corresponding to the red
      highlighted dots in \cref{kphf-tuning} that target $1.5\times$ the space
      lower bound for each $k$ and load factor. The \emph{final} load factor and space usage
      take into account the additional slots and space used for bumping.
      The time for construction is measured with a single thread, and the query
      time is measured as the average time per iteration of a for loop.}
	\begin{tabular}{ccrcclcc}
		\toprule
		\multicolumn{1}{c}{$k$} & \multicolumn{2}{c}{Load factor $\alpha$} 
        & \multicolumn{2}{c}{Space (bits/key)}
        & \multicolumn{1}{c}{Bumped} & \multicolumn{1}{c}{Build time} & \multicolumn{1}{c}{Query time}  \\
		& \multicolumn{1}{c}{target} & \multicolumn{1}{c}{final}
        & \multicolumn{1}{c}{target} &\multicolumn{1}{c}{final} &
        \multicolumn{1}{c}{($\%$)} & \multicolumn{1}{c}{(ns/key)} & \multicolumn{1}{c}{(ns/key)}  \\
		\toprule
4&   0.70  &    $\approx$ 0.700 &  0.204 &        0.204  &   0.001  &191  &3.2\\
4&   0.80  &    $\approx$ 0.800 &  0.316 &        0.316  &   0.001  &190  &3.4\\
4&   0.90  &    $\approx$ 0.900 &  0.489 &        0.490  &   0.005  &206  &3.7\\
4&   0.99  &            0.980 &  0.800 &        0.808  &   0.53  &277  &3.9\\
\midrule
8&   0.70  &    $\approx$ 0.700 &  0.065 &        0.065  &   0.0002  &398  &2.6\\
8&   0.80  &    $\approx$ 0.800 &  0.124 &        0.124  &   0.0001  &473  &3.0\\
8&   0.90  &    $\approx$ 0.900 &  0.231 &        0.231  &   0.002  &450  &3.2\\
8&   0.99  &            0.976 &  0.460 &        0.467  &   0.72  &500  &3.6\\
\midrule
16&  0.70  &    $\approx$ 0.700 &  0.015 &        0.015  &   0.13  &450  &2.5\\
16&  0.80  &            0.798 &  0.040 &        0.041  &   0.11  &480  &2.6\\
16&  0.90  &            0.896 &  0.097 &        0.098  &   0.23  &443  &2.9\\
16&  0.99  &            0.968 &  0.252 &        0.258  &   1.12  &499  &3.4\\
		\bottomrule
	\end{tabular}
\end{table}

\parag{Comparison to previous $k$-PHFs.}
Similar to the threshold-based \emph{minimal} $k$-PHF with consensus developed in \cite{hermann2025engineering}, we are able to
achieve roughly \(1.5\times\) the space lower bound for \(k\approx 10\).
At \(\approx 200\) ns/key (also for \(n=10^8\)), our construction also has a similar construction speed.
Although the query performance is not directly comparable due to string hashing,
our \(\approx 4\) ns/query (or 250 Mqueries/s) is much faster than even the
fastest threshold-based scheme at \(\approx 50\) Mqueries/s.
The $k$-RecSplit method gets as low as \(1.1\times\) space overhead, but is much slower to build and query.

CHD \cite{belazzougui2009hash} achieves a space usage of 1.03, 0.77, and 0.60
bits per key for $k\in\{4, 8, 16\}$ respectively when using a load factor of
$\alpha=0.99$. $k$-PtrHash uses less space at 0.80, 0.47, and 0.25 bits per key
respectively, but bumps around $1\%$ of the keys and thus has a smaller overall
load factor around $0.98$.

\parag{Varying $k$.}
While the current implementation is tuned towards \(4\leq k\leq 16\),
the bumping fallback ensures that construction will always terminate, even for larger \(k\), although not always as space-efficiently.
The \(x^4\) bucket function prevents it from working directly for \(k=1\) due to
self-collisions in the first (largest) bucket, but changing that
to \(x^2\) gives results comparable to PtrHash when \(\alpha\) is not too
close to \(1\). For example, for \(\alpha = 0.9\) and targeting \(1.5\times\) space
overhead, we get \(2.1\%\) bumped keys and a final load factor of \(0.87\) at \(1.67\)
bits per key.

\section{Application: Static \texorpdfstring{$\bm k$}{k}-PHF-sets}\label{s:application}
In this section, we will use our non-minimal $k$-PHF to develop a fast and memory
efficient static hash set. 
We first review some existing (static) hash sets, and then present our own data structure
and experiments. For simplicity, we assume 64-bit keys and no associated values.
We note that there have been many previous reviews on dynamic hash tables
\cite{hash-table-bench,hash-table-bench-2,hash-table-bench-3}.

\subsection{Previous work on static hash sets} \label{sec:prev-work}
\enlargethispage{2em}
\parag{SwissTable.}
A commonly used \emph{dynamic} hash map implementation is Google's SwissTable
\texttt{absl::flat\_hash\_map} \cite{abseilswisstables},
which is also the basis for Rust's standard \texttt{HashMap}, which is based on
the \texttt{hashbrown} crate.
It stores an array of \(n/\alpha\) slots with load factor \(7/16 \leq \alpha \leq 7/8\),
and additionally a metadata array of 7-bit \emph{fingerprints} (and 1 tombstone bit
marking deletions) for each slot.
A hash function is used to map a query key to a group of 8 slots (or more, when
using an implementation with SIMD).
Its fingerprint is then compared, and quadratic probing to the next group
happens until either a matching fingerprint or an empty slot is found.
Once a fingerprint matches, the query is checked against the corresponding key.

The fingerprints act as a small (and thus
fast) approximate filter for negative queries.
On average, \(\Omega(1/\epsilon)\) probes are needed per query (where \(\epsilon =
1-\alpha\)) before the key or an empty/deleted slot is found.
Efficient prefetching is possible if and insofar as the first probe is sufficient for answering a query
A drawback is that
the table only supports sizes that are powers of two, causing
worst-case \(2.8\times\) space overhead compared to just an array of keys.

\parag{Static cache line-based hash set.}
A simple implementation of a static hash set is the following:
Allocate \(\lceil n/\alpha /8\rceil\) 512-bit cache line sized bins that each fit
8 64-bit keys.
Then hash each key to a cache line, and use linear probing until a bin is found
that still has capacity, and insert it there.
Empty slots are represented by \(0\), and \(0\) itself is handled separately by
simply storing a boolean indicating whether $0$ is present or not.
To query a key, the key is compared (using SIMD instructions) against all values
in a cache line, with linear probing until either the key or an empty slot is found.

This has the same drawbacks of probing as SwissTable does, but avoids the
indirection to the fingerprint metadata. This is especially beneficial
in cases with little probing, i.e., for positive queries when the load factor is small.

\parag{Blocked cuckoo hashing.}
Blocked Cuckoo-hashing
\cite{pagh2004cuckoo,fountoulakis2016multiple,dietzfelbinger2010tight}
avoids unbounded probing by selecting exactly two candidate cache line sized blocks
(or windows \cite{walzer2017load,windowed-cuckoo-hashing-blog}) for each
key. When both blocks are already full, one of the current keys in them is evicted and a BFS or
DFS is done to displace them.
To query whether a key is present, one can either \emph{eagerly} load both cache lines from
memory at the same time (doubling the required memory bandwidth) or \emph{lazily} load the
second only when the key is not in the first (increasing latency in case a key
is not present in its primary location).

\parag{$1$-PHF-set.} Instead of probing,
one can also use a $1$-PHF to map keys to slots without collisions \cite{sprugnoli1977perfect, lu2006perfect}.
FPH \cite{fph-table} is a fast implementation that uses 
a non-minimal $1$-PHF based on FCH \cite{fox1992faster}.
This PHF directly maps each key to a single slot, which stores the value of the
key mapping there.
As long as the PHF fits in cache, this allows for fast single-cache-miss queries,
with no performance distinction between positive and negative queries.
The same is possible by using PtrHash or PHast, which have less space overhead
and higher query throughput than FCH.

\parag{Previous $k$-PHF approaches.}
The cache line-based approach and $1$-PHF approach can be merged by using a
$k$-PHF to map keys to cache lines, where $k$ is the number of keys per cache
line.
For example one can use the non-minimal $k$-PHF CHD \cite{belazzougui2009hash} or the \emph{minimal} $k$-PHFs developed in \cite{hermann2025engineering}.
Another line of work considers dynamic hash tables that guarantee a single cache miss and thus implicitly contain a
larger \emph{dynamic} $k$-PHF.
Notable techniques are separator hashing \cite{gonnet1988external, larson1984file, larson1988linear}, MapEmbed on CPU/FPGA
\cite{wu2021mapembed}, external Robin-Hood hashing \cite{celis1988external}, EEPH for persistent memory \cite{chen2023eeph}, or GPH
for GPUs \cite{cao2025gph}.
We remark that the extendible hashing technique \cite{fagin1979extendible, pagh2003basic} has up to two cache misses.

\subsection{The static \texorpdfstring{$\bm k$}{k}-PHF-set}
We extend the $1$-PHF-sets based on $1$-PHF to a static $k$-PHF-set
based on $k$-PtrHash.
This enables a high-load-factor hash set that only needs to read exactly one cache
line from RAM as long as the $k$-PHF fits in memory.
As in the previous methods, we treat each cache line as a bin with \(k=8\)
slots, and use SIMD instructions to check if a key is contained in a cache line.
The main benefit is that a static $k$-PHF needs much less space than a static
$1$-PHF necessarily needs, and much less than a dynamic $k$-PHF needs in practice,
allowing a fixed-size cache to store a PHF for a larger number of keys.
Starting with minimal 1-PHFs which needs $1.44$ bits/key,
going to $\alpha=0.9$ requires $1.07$ bits/key, and then increasing $k$ to $8$
lowers the space lower bound to $0.15$ bits/key, $7\times$ less than for $k=1$.
Thus, a smaller $\alpha$ makes the $k$-PHF-set more cache-friendly and faster,
even though more memory is used.


\begin{figure}[h]
\centering
\makebox[\linewidth]{
\includegraphics[width=1.0\linewidth]{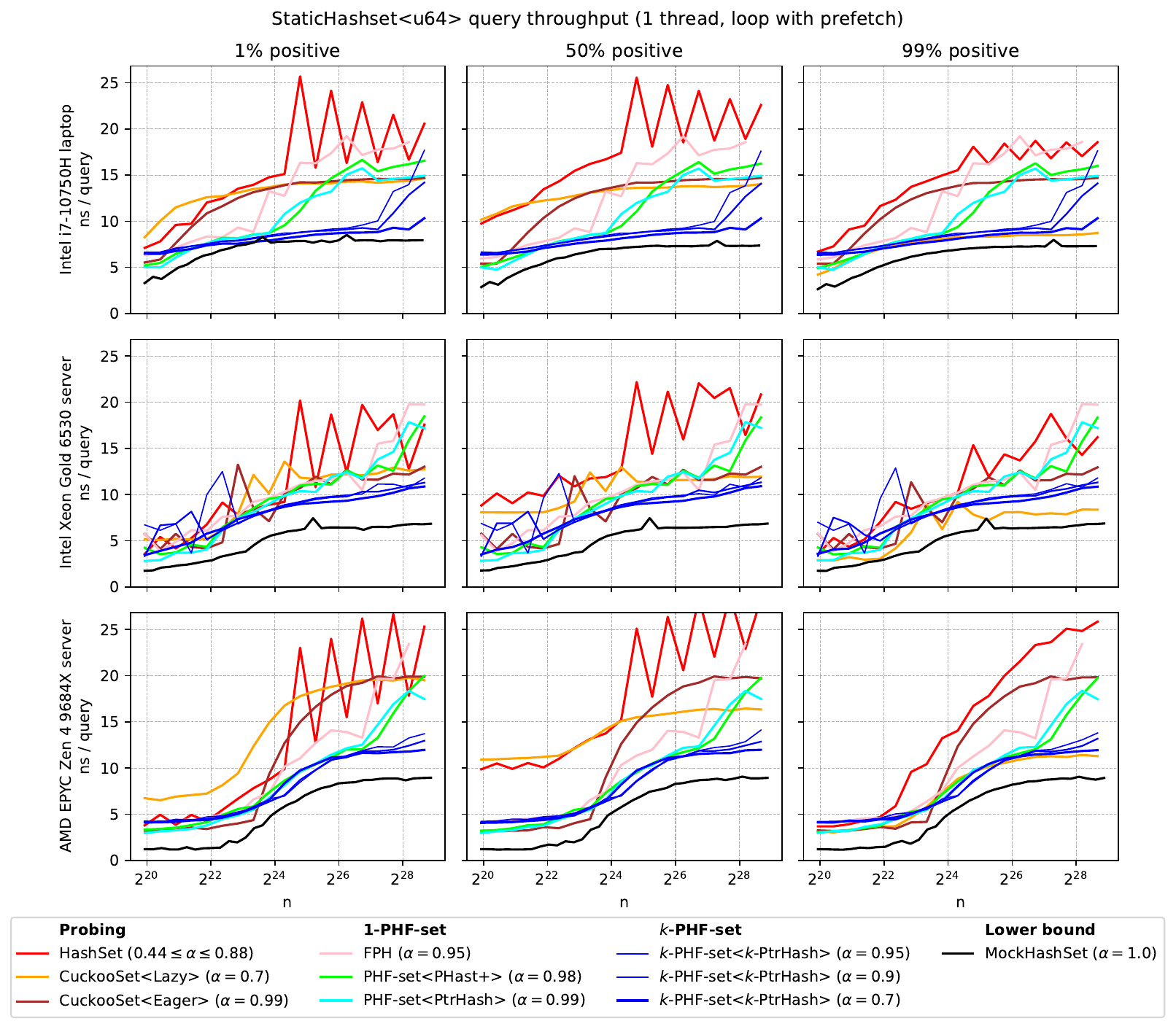}
}
\caption{\label{set-results}Query throughput of various (static) hash set
  implementations for \(n\) 64-bit keys. RAM accesses are prefetched 32
  loop iterations ahead.
  Colours indicate different data structures, and for $k$-PHF-set, the line
  width indicates the load factor from \(\alpha=0.7\) (thick) to \(\alpha=0.95\)
  (thin).
  The black line indicates the throughput of a mock implementation that checks
  if the key equals the first word in a random cache line, and thus gives an
  indication of the maximum achievable throughput of a hash set that does not
  use an internal filter.
  The left/middle/right column show throughput for \(1\%\) / \(50\%\) /
  \(99\%\) positive queries. The top/middle/bottom are for an Intel laptop, Intel
  Xeon server, and AMD EPYC server.
\ifarxiv
    \cref{app:results}
\else
    The extended version \cite{extended-version}
\fi
   contains further results on multiple threads and without prefetching.}
\end{figure}

\subsection{Results}
\parag{Competitors.}
We compare our $k$-PHF-set based on $k$-PtrHash against the methods
mentioned in \cref{sec:prev-work}.
The first is the Rust \texttt{HashSet} in \texttt{hashbrown}.
We implemented the blocked cuckoo hashing ourselves as a known baseline using the same SIMD instructions as
the version based on $k$-PtrHash. (We omit the linear-probing variant as it is
consistently slower.)
We call \texttt{FphDynSet} \cite{fph-table} through Rust bindings and also added
support for prefetching the slot of a key.
We wrap the fastest configurations of PHast+ and PtrHash in our own
$1$-PHF-set. All Rust implementations use GxHash \cite{gxhash} as the primary
hash function.

We also tested the $k$-PHF-set with the fastest bucket placement and threshold based bumping
minimal $k$-PHF techniques of \cite{hermann2025engineering}, but both were
$\approx 2$ times slower than the standard Rust HashSet. Similarly, MapEmbed
\cite{wu2021mapembed} was $\approx 5$ times slower, while we could not find any
source code for separator hashing \cite{gonnet1988external, larson1984file, larson1988linear},
external Robin-Hood hashing \cite{celis1988external},
extendible hashing \cite{fagin1979extendible} and
EEPH \cite{chen2023eeph}.
The query time of the $k$-PHF technique CHD \cite{belazzougui2009hash} is slower compared to PtrHash by an order of magnitude.
Thus, these methods were excluded.

For each static method, we show the best load factor. For the linear probing based hash set and the lazy
cuckoo set, we use $\alpha=0.7$ to reduce probing costs. For the eager cuckoo
set, we use $\alpha=0.99$ since we always read both locations anyway.
While FPH works for $\alpha \leq 0.98$, we choose $\alpha=0.9$ to avoid
excessively slow construction on large inputs.
PHast+ and PtrHash without remapping naturally have a load factor around $0.98$
and $0.99$ respectively.
For $k$-PHF-Set with $k$-PtrHash, we test load factors $\alpha\in\{0.7,0.9,0.95\}$.


\parag{Setup.}
As our data structure is designed to be I/O-efficient, we benchmark it in a
\emph{high throughput} setting: we do 3 million queries, and \emph{prefetch} the
cache line for the query 32 iterations ahead. Specifically, the $1$-PHF-set and
$k$-PHF-set compute the index of the slot/bin a key maps to and prefetch this
cache line. This index is then stored, and the corresponding data is read 32 iterations later
to test whether the key is indeed present, so that the key does not have to be
hashed again.
We forked \texttt{hashbrown} and FPH to add similar support for prefetching and
continuing the query some time later.
For cuckoo hashing, we test two modes: one prefetching only the first
(preferred) cache line, and one where we eagerly fetch the second location as well.

Additionally, we test a mock implementation of a basic hash set based on cache
lines. This maps each random 64-bit key to one of the $b$ cache lines via a single
multiplication using $\lfloor k\cdot b / 2^{64} \rfloor$ and then ``checks'' if
the key is in the hash set by testing if it equals the first word in the cache line.
This indicates the maximum achievable throughput of a hash set that does not use
an internal filter.

We run the benchmarks on an Intel Skylake i7-10750H laptop (6 cores, 32
KiB L1 per core, 256 KiB L2 per core, 12 MiB L3 cache), an Intel Xeon Gold 6530 server (64 cores,
128 threads, 32 KiB L1 per core, 2 MiB L2 per core, 160 MiB L3 cache per 32
cores), and an AMD EPYC Zen 4 9684X server (96 cores, 192 threads, 64 KiB L1 per
core, 1 MiB L2 per core, 32+64 MiB L3 per 8 cores).

\ifarxiv
    \cref{app:results}
\else
    The extended version of this paper \cite{extended-version}
\fi
contains additional results without prefetching and with multithreading.

\parag{Discussion.}
Results are shown in \cref{set-results}.
The dynamic \texttt{HashSet} (red) is relatively slow, and its speed depends heavily
on the load factor, which varies by $n$ due to the power-of-$2$ restriction on the 
underlying buffer.
(We note that in the additional results
\ifarxiv
    in \cref{app:results}
\else
    in the extended version \cite{extended-version}
\fi
there are
cases where it is the fastest for negative queries, due to the $7$-bit filter.
Indeed, this allows it to be faster than the lower bound given by our mock implementation.)
The lazy cuckoo table (orange) is fast for positive queries since (due to the small load
factor) these require nearly no probing, but in other cases it is as slow as the eager variant
(brown) that has double the memory bandwidth. The three
$1$-PHF-sets (pink, lime, cyan) all perform similar and independent of the type
of queries, with FPH being slightly
slower than using PHast+ or PtrHash. Specifically, they remain fast up to
$n\approx 2^{24}$, whereas the probing methods slow down at $n\approx 2^{21}$.

As anticipated, using $k$-PtrHash indeed allows larger $n$ up to $2^{27}$ (on
the laptop) before
the PHF does not fit in cache anymore and queries slow down, compared to
$n=2^{24}$ for the 1-PHF-set, resulting in
up to $1.6\times$ higher throughput in this range.
Decreasing $\alpha$
from $0.95$ to $0.7$ maintains high throughput up to $2^{29}$, as the $k$-PHF takes less space for
smaller $\alpha$.

The two server machines have larger caches, and our $k$-PHF-set only starts to
be faster than the 1-PHF-set at $n=2^{27}$ here. Due to the different memory
characteristics of these machines, the cuckoo table is generally competitive on
the Intel Xeon machine, but slow on the AMD EPYC, where the $k$-PHF-set has
up to $1.5\times$ higher throughput than other methods for negative queries.
Overall, our $k$-PHF-set is the only method that is competitive on all machines
for both positive and negative queries.


\section{Conclusion}
\enlargethispage{2em}
We have given a tight lower bound on the space usage of $k$-PHFs with load factor
$0<\alpha< 1$ (\cref{thm:precise-formula}).
As the formula for the exact space usage is somewhat unwieldy, we also derive explicit simpler formulas that neglect constant factors.
We extended PtrHash to $k$-perfect hashing with two modifications:
we bump problematic buckets rather than using evictions,
and we use a scoring function to preferentially select seeds that result in more balanced bins.
For $4\leq k\leq 16$ and $0.7\leq \alpha\leq 0.9$, this
scheme works down to $50\%$ space overhead.

An $8$-PHF-based hash set that maps 64-bit keys to cache line-sized
bins has as high or higher throughput than other hash sets, and depending on the
hardware has up to $1.5\times$ higher throughput in cases where a $1$-PHF on the
data does not fit in cache.

\parag{Future work.}
It appears that none of the current highly space-efficient strategies for constructing PHFs generalize well to the non-minimal $k$-PHF setting, and it remains open to find methods that are close to the space lower bound.

Concerning $k$-PtrHash, much remains unclear on e.g. the
optimal combination of bucket function and scoring function to use.
%
%
From the practical side, more work could be done to speed up the $k$-PtrHash
construction using e.g. SIMD instructions, and to extend $k$-PHF-set into a
map with updatable values.
Additionally, it might be possible to use quotienting to
store only a subset of the bits of each key, thus allowing a larger $k$.









\bibliographystyle{plainurl}
\bibliography{paper}

@article {deacon,
  author = {Constantinides, Bede and Lees, John and Crook, Derrick W},
  title = {Deacon: fast sequence filtering and contaminant depletion},
  elocation-id = {2025.06.09.658732},
  year = {2025},
  doi = {10.1101/2025.06.09.658732},
  journal = {bioRxiv}
}

@Misc{abseilswisstables,
  author = {Google},
  title = {{Abseil's Swiss Table}},
  howpublished = {\url{https://abseil.io/about/design/swisstables}},
  year = {2017}
}

@InProceedings{hermann2025engineering,
  author       = {Hermann, Stefan and Kirmayer, Sebastian and Lehmann,
                  Hans-Peter and Sanders, Peter and Walzer, Stefan},
  title        = {Engineering Minimal $k$-Perfect Hash Functions},
  year         = 2025,
  language     = {en},
  volume       = 351,
  publisher    = {Schloss Dagstuhl – Leibniz-Zentrum für Informatik},
  pages        = {99:1-99:18},
  doi          = {10.4230/LIPIcs.ESA.2025.99},
  booktitle    = {{ESA}},
  copyright    = {Creative Commons Attribution 4.0 International license}
}

@inproceedings{belazzougui2009hash,
  author = {Djamal Belazzougui and
    Fabiano C. Botelho and
    Martin Dietzfelbinger},
  title = {Hash, Displace, and Compress},
  booktitle = {{ESA}},
  series = {LNCS},
  volume = {5757},
  pages = {682--693},
  publisher = {Springer},
  year = {2009},
  doi = {10.1007/978-3-642-04128-0_61}
}

@inproceedings{pibiri2021pthash,
  author = {Giulio Ermanno Pibiri and
    Roberto Trani},
  title = {{PTHash}: Revisiting {FCH} Minimal Perfect Hashing},
  booktitle = {{SIGIR}},
  pages = {1339--1348},
  publisher = {{ACM}},
  year = {2021},
  doi = {10.1145/3404835.3462849}
}

@inproceedings{fox1992faster,
  author =       {Fox, Edward A. and Chen, Qi Fan and Heath, Lenwood S.},
  title =        {A faster algorithm for constructing minimal perfect hash functions},
  booktitle =      {{SIGIR}},
  year =         1992,
  doi =          {10.1145/133160.133209},
  publisher =    {ACM Press}
}

@inproceedings{limasset2017fast,
  author = {Antoine Limasset and
    Guillaume Rizk and
    Rayan Chikhi and
    Pierre Peterlongo},
  title = {Fast and Scalable Minimal Perfect Hashing for Massive Key Sets},
  booktitle = {{SEA}},
  series = {LIPIcs},
  volume = {75},
  pages = {25:1--25:16},
  year = {2017},
  doi = {10.4230/LIPIcs.SEA.2017.25}
}

@inproceedings{esposito2020recsplit,
  author = {Emmanuel Esposito and
    Thomas Mueller Graf and
    Sebastiano Vigna},
  title = {{RecSplit}: Minimal Perfect Hashing via Recursive Splitting},
  booktitle = {{ALENEX}},
  pages = {175--185},
  publisher = {{SIAM}},
  year = {2020},
  doi = {10.1137/1.9781611976007.14}
}

@inproceedings{genuzio2016fast,
  author = {Marco Genuzio and
    Giuseppe Ottaviano and
    Sebastiano Vigna},
  title = {Fast Scalable Construction of (Minimal Perfect Hash) Functions},
  booktitle = {{SEA}},
  series = {LNCS},
  volume = {9685},
  pages = {339--352},
  publisher = {Springer},
  year = {2016},
  doi = {10.1007/978-3-319-38851-9_23}
}

@inproceedings{kurpicz2023pachash,
  author = {Florian Kurpicz and
    Hans-Peter Lehmann and
    Peter Sanders},
  title = {{PaCHash}: Packed and Compressed Hash Tables},
  booktitle = {{ALENEX}},
  pages = {162--175},
  publisher = {{SIAM}},
  year = {2023},
  doi = {10.1137/1.9781611977561.CH14}
}

@article{chapman2011meraculous,
  author = {Chapman, Jarrod A. and others},
  doi = {10.1371/journal.pone.0023501},
  journal = {PLOS ONE},
  month = {08},
  number = {8},
  pages = {1-13},
  title = {Meraculous: De Novo Genome Assembly with Short Paired-End Reads},
  volume = {6},
  year = {2011}
}

@Article{lehmann2025modern,
  author       = {Lehmann, Hans-Peter and Mueller, Thomas and Pagh, Rasmus and
                  Pibiri, Giulio Ermanno and Sanders, Peter and Vigna,
                  Sebastiano and Walzer, Stefan},
  title        = {Modern Minimal Perfect Hashing: A Survey},
  journal      = {ACM Computing Surveys},
  year         = 2026,
  volume       = 58,
  number       = 10,
  month        = Mar,
  pages        = {1–36},
  issn         = {1557-7341},
  doi          = {10.1145/3797036},
  publisher    = {Association for Computing Machinery (ACM)}
}

@inproceedings{bez2023high,
  author = {Dominik Bez and others},
  title = {High Performance Construction of {RecSplit} Based Minimal Perfect Hash
    Functions},
  booktitle = {{ESA}},
  series = {LIPIcs},
  volume = {274},
  pages = {19:1--19:16},
  year = {2023},
  doi = {10.4230/LIPIcs.ESA.2023.19}
}

@article{beling2023fingerprinting,
  author = {Piotr Beling},
  title = {Fingerprinting-based Minimal Perfect Hashing Revisited},
  journal = {{ACM} J. Exp. Algorithmics},
  volume = {28},
  pages = {1.4:1--1.4:16},
  year = {2023},
  doi = {10.1145/3596453}
}

@article{czech1992optimal,
  author = {Zbigniew J. Czech and
    George Havas and
    Bohdan S. Majewski},
  title = {An Optimal Algorithm for Generating Minimal Perfect Hash Functions},
  journal = {Inf. Process. Lett.},
  volume = {43},
  number = {5},
  pages = {257--264},
  year = {1992},
  doi = {10.1016/0020-0190(92)90220-P}
}

@article{botelho2004new,
  title = {A new algorithm for constructing minimal perfect hash functions},
  author = {Botelho, Fabiano C and Gomes, David M and Ziviani, Nivio},
  journal = {Preprint},
  year = {2004},
}

@article{majewski1996family,
  author = {Bohdan S. Majewski and
    Nicholas C. Wormald and
    George Havas and
    Zbigniew J. Czech},
  title = {A Family of Perfect Hashing Methods},
  journal = {Comput. J.},
  volume = {39},
  number = {6},
  pages = {547--554},
  year = {1996},
  doi = {10.1093/COMJNL/39.6.547}
}

@inproceedings{mehlhorn1982program,
  author = {Kurt Mehlhorn},
  title = {On the Program Size of Perfect and Universal Hash Functions},
  booktitle = {{FOCS}},
  pages = {170--175},
  publisher = {{IEEE}},
  year = {1982},
  doi = {10.1109/SFCS.1982.80}
}

@inproceedings{weaver2020constructing,
  author = {Sean A. Weaver and
    Marijn Heule},
  title = {Constructing Minimal Perfect Hash Functions Using {SAT} Technology},
  booktitle = {{AAAI}},
  pages = {1668--1675},
  publisher = {{AAAI} Press},
  year = {2020},
  doi = {10.1609/AAAI.V34I02.5529}
}

@inproceedings{lehmann2023sichash,
  author = {Hans-Peter Lehmann and
    Peter Sanders and
    Stefan Walzer},
  title = {{SicHash} -- Small Irregular Cuckoo Tables for Perfect Hashing},
  booktitle = {{ALENEX}},
  pages = {176--189},
  publisher = {{SIAM}},
  year = {2023},
  doi = {10.1137/1.9781611977561.CH15}
}

@article{cao2025gph,
  author       = {Cao, Jiaping and Xu, Le and Yiu, Man Lung and Qin, Jianbin and Tang, Bo},
  title        = {{GPH: An Efficient and Effective Perfect Hashing Scheme for GPU Architectures}},
  journal      = {Proceedings of the ACM on Management of Data},
  year         = 2025,
  volume       = 3,
  number       = 3,
  month        = {June},
  pages        = {1–26},
  issn         = {2836-6573},
  doi          = {10.1145/3725406},
  publisher    = {Association for Computing Machinery (ACM)}
}

@inproceedings{wu2021mapembed,
  author       = {Wu, Yuhan and Liu, Zirui and Yu, Xiang and Gui, Jie and Gan,
                  Haochen and Han, Yuhao and Li, Tao and Rottenstreich, Ori and
                  Yang, Tong},
  title        = {MapEmbed: Perfect Hashing with High Load Factor and Fast
                  Update},
  year         = 2021,
  booktitle    = {Proceedings of the 27th ACM SIGKDD Conference on Knowledge
                  Discovery \& Data Mining},
  series       = {KDD ’21},
  publisher    = {ACM},
  month        = Aug,
  pages        = {1863–1872},
  doi          = {10.1145/3447548.3467240},
  collection   = {KDD ’21}
}

@InProceedings{chen2023eeph,
  author       = {Chen, Qi and Hu, Hao and Deng, Cai and Liu, Dingbang and Li,
                  Shiyi and Tang, Bo and Yao, Ting and Xia, Wen},
  title        = {{EEPH: An Efficient Extendible Perfect Hashing for Hybrid
                  PMem-DRAM}},
  year         = 2023,
  booktitle    = {2023 IEEE 39th International Conference on Data Engineering
                  (ICDE)},
  publisher    = {IEEE},
  month        = Apr,
  pages        = {1366–1378},
  doi          = {10.1109/icde55515.2023.00109},
}

@article{sprugnoli1977perfect,
  author = {Renzo Sprugnoli},
  title = {Perfect Hashing Functions: {A} Single Probe Retrieving Method for
    Static Sets},
  journal = {Commun. {ACM}},
  volume = {20},
  number = {11},
  pages = {841--850},
  year = {1977},
  doi = {10.1145/359863.359887}
}

@inproceedings{lehmann2023shockhash,
  author = {Hans-Peter Lehmann and
    Peter Sanders and
    Stefan Walzer},
  title = {{ShockHash}: Towards Optimal-Space Minimal Perfect Hashing Beyond Brute-Force},
  booktitle = {{ALENEX}},
  pages = {194--206},
  publisher = {{SIAM}},
  year = {2024},
  doi = {10.1137/1.9781611977929.15}
}

@inproceedings{chazelle2004bloomier,
  author = {Bernard Chazelle and
    Joe Kilian and
    Ronitt Rubinfeld and
    Ayellet Tal},
  title = {The Bloomier filter: an efficient data structure for static support
    lookup tables},
  booktitle = {{SODA}},
  pages = {30--39},
  publisher = {{SIAM}},
  url = {http://dl.acm.org/citation.cfm?id=982792.982797},
  year = {2004}
}

@inproceedings{botelho2007external,
  author = {Fabiano C. Botelho and
    Nivio Ziviani},
  title = {External perfect hashing for very large key sets},
  booktitle = {{CIKM}},
  pages = {653--662},
  publisher = {{ACM}},
  year = {2007},
  doi = {10.1145/1321440.1321532}
}

@inproceedings{pagh2003basic,
  author = {Rasmus Pagh},
  title = {Basic External Memory Data Structures},
  booktitle = {Algorithms for Memory Hierarchies},
  series = {LNCS},
  volume = {2625},
  pages = {14--35},
  publisher = {Springer},
  year = {2003},
  doi = {10.1007/3-540-36574-5_2}
}

@article{pagh2004cuckoo,
  author = {Rasmus Pagh and
    Flemming Friche Rodler},
  title = {Cuckoo hashing},
  journal = {J. Algorithms},
  volume = {51},
  number = {2},
  pages = {122--144},
  year = {2004},
  doi = {10.1016/j.jalgor.2003.12.002}
}

@article{larson1988linear,
  author = {Per{-}{\AA}ke Larson},
  title = {Linear Hashing with Separators - {A} Dynamic Hashing Scheme Achieving
    One-Access Retrieval},
  journal = {{ACM} Trans. Database Syst.},
  volume = {13},
  number = {3},
  pages = {366--388},
  year = {1988},
  doi = {10.1145/44498.44500}
}

@article{larson1984file,
  author = {Per{-}{\AA}ke Larson and
    Ajay Kajla},
  title = {File Organization: Implementation of a Method Guaranteeing Retrieval
    in One Access},
  journal = {Commun. {ACM}},
  volume = {27},
  number = {7},
  pages = {670--677},
  year = {1984},
  doi = {10.1145/358105.358193}
}

@article{gonnet1988external,
  author = {Gaston H. Gonnet and
    Per{-}{\AA}ke Larson},
  title = {External hashing with limited internal storage},
  journal = {J. {ACM}},
  volume = {35},
  number = {1},
  pages = {161--184},
  year = {1988},
  doi = {10.1145/42267.42274}
}

@phdthesis{celis1988external,
  title={External robin hood hashing},
  author={Celis, Pedro},
  school = {University of Waterloo},
  year={1988}
}

@article{fagin1979extendible,
  author = {Ronald Fagin and
    J{\"{u}}rg Nievergelt and
    Nicholas Pippenger and
    H. Raymond Strong},
  title = {Extendible Hashing - {A} Fast Access Method for Dynamic Files},
  journal = {{ACM} Trans. Database Syst.},
  volume = {4},
  number = {3},
  pages = {315--344},
  year = {1979},
  doi = {10.1145/320083.320092}
}

@inproceedings{dietzfelbinger2010tight,
  author = {Martin Dietzfelbinger and others},
  title = {Tight Thresholds for Cuckoo Hashing via {XORSAT}},
  booktitle = {{ICALP}},
  series = {LNCS},
  volume = {6198},
  pages = {213--225},
  publisher = {Springer},
  year = {2010},
  doi = {10.1007/978-3-642-14165-2_19}
}

@inproceedings{beling2026phast,
author = {Piotr Beling and Peter Sanders},
title = {PHast – Perfect Hashing made fast},
booktitle = {{ALENEX}},
year = {2026},
pages = {1-14},
doi = {10.1137/1.9781611978957.1},
}

@article{walzer2017load,
  author       = {Walzer, Stefan},
  title        = {Load Thresholds for Cuckoo Hashing with Overlapping Blocks},
  journal      = {ACM Transactions on Algorithms},
  year         = 2023,
  volume       = 19,
  number       = 3,
  month        = May,
  pages        = {1–22},
  issn         = {1549-6333},
  doi          = {10.1145/3589558},
  publisher    = {Association for Computing Machinery (ACM)}
}

@InProceedings{lu2006perfect,
  author       = {Lu, Yi and Prabhakar, Balaji and Bonomi, Flavio},
  title        = {Perfect Hashing for Network Applications},
  year         = 2006,
  booktitle    = {2006 IEEE International Symposium on Information Theory},
  publisher    = {IEEE},
  month        = {July},
  pages        = {2774–2778},
  doi          = {10.1109/isit.2006.261567}
}

@inproceedings{hermann2024phobic,
  author = {Stefan Hermann and others},
  title = {{PHOBIC:} Perfect Hashing With Optimized Bucket Sizes and Interleaved
    Coding},
  booktitle = {{ESA}},
  series = {LIPIcs},
  volume = {308},
  pages = {69:1--69:17},
  year = {2024},
  doi = {10.4230/LIPIcs.ESA.2024.69}
}

@article{joag1983negative,
  author = {Kumar Joag-Dev and Frank Proschan},
  title = {Negative Association of Random Variables with Applications},
  volume = {11},
  journal = {The Annals of Statistics},
  number = {1},
  publisher = {Institute of Mathematical Statistics},
  pages = {286--295},
  year = {1983},
  doi = {10.1214/aos/1176346079}
}

@article{fountoulakis2016multiple,
  author = {Nikolaos Fountoulakis and
    Megha Khosla and
    Konstantinos Panagiotou},
  title = {The Multiple-Orientability Thresholds for Random Hypergraphs},
  journal = {Comb. Probab. Comput.},
  volume = {25},
  number = {6},
  pages = {870--908},
  year = {2016},
  doi = {10.1017/S0963548315000334}
}

@inproceedings{mairson1983program,
  author = {Harry G. Mairson},
  title = {The Program Complexity of Searching a Table},
  booktitle = {{FOCS}},
  pages = {40--47},
  publisher = {{IEEE}},
  year = {1983},
  doi = {10.1109/SFCS.1983.76}
}

@phdthesis{lehmann2024fast,
  author = {Lehmann, Hans-Peter},
  year = {2024},
  title = {Fast and Space-Efficient Perfect Hashing},
  doi = {10.5445/IR/1000176432},
  publisher = {{Karlsruher Institut für Technologie (KIT)}},
  pagetotal = {167},
  school = {Karlsruher Institut für Technologie (KIT)},
  language = {english}
}

@article{lehmann2025consensus,
  author       = {Lehmann, Hans-Peter and Sanders, Peter and Walzer, Stefan and
                  Ziegler, Jonatan},
  title        = {Combined Search and Encoding for Seeds, with an Application to
                  Minimal Perfect Hashing},
  year         = 2025,
  language     = {en},
  volume       = 351,
  publisher    = {Schloss Dagstuhl – Leibniz-Zentrum für Informatik},
  pages        = {109:1-109:18},
  doi          = {10.4230/LIPIcs.ESA.2025.109},
  journal      = {LIPIcs, Volume 351, ESA 2025},
  copyright    = {Creative Commons Attribution 4.0 International license}
}

@InProceedings{grootkoerkamp2025ptrhash,
  author =  {Groot Koerkamp, Ragnar},
  title = {{PtrHash: Minimal Perfect Hashing at RAM Throughput}},
  booktitle = {{SEA}},
  pages = {21:1--21:21},
  series =  {{LIPIcs}},
  year =  {2025},
  volume =  {338},
  doi =   {10.4230/LIPIcs.SEA.2025.21},
}

@article{hermann2025morphishash,
  author       = {Hermann, Stefan},
  title        = {MorphisHash: Improving Space Efficiency of ShockHash for
                  Minimal Perfect Hashing},
  year         = 2025,
  language     = {en},
  volume       = 351,
  publisher    = {Schloss Dagstuhl – Leibniz-Zentrum für Informatik},
  pages        = {9:1-9:16},
  doi          = {10.4230/LIPIcs.ESA.2025.9},
  journal      = {LIPIcs, Volume 351, ESA 2025},
  copyright    = {Creative Commons Attribution 4.0 International license}
}

@Article{classical-local-limit-theorems,
  author       = {Szewczak, Zbigniew and Weber, Michel},
  title        = {Classical and Almost Sure Local Limit Theorems},
  year         = 2022,
  doi          = {10.48550/ARXIV.2208.02700},
  journal    = {arXiv},
  copyright    = {arXiv.org perpetual, non-exclusive license}
}

@Article{local-limit-theorem,
  author       = {B. V. Gnedenko},
  title        = {On a local limit theorem of the theory of probability},
  journal      = {Uspekhi Mat. Nauk},
  year         = 1948,
  volume       = 3,
  pages        = {187-194},
  url          = {http://mi.mathnet.ru/eng/rm8713},
}

@Misc{windowed-cuckoo-hashing-blog,
  author = {Reiner Pope},
  title = {{Cuckoo hashing improves SIMD hash tables (and other hash table tradeoffs)}},
  howpublished = {\url{https://reiner.org/cuckoo-hashing}},
  year = {2025}
}

@Misc{fph-table,
  author = {Ren Zibei},
  title = {{Flash Perfect Hash}},
  howpublished = {\url{https://github.com/renzibei/fph-table}},
  year = {2022}
}

@Misc{hash-table-bench-3,
  author = {Ren Zibei},
  title = {{Hash Table Benchmark}},
  howpublished = {\url{https://github.com/renzibei/hashtable-bench}},
  year = {2022}
}

@Misc{hash-table-bench,
  author = {Martin Leitner-Ankerl},
  title = {{Comprehensive C++ Hashmap Benchmarks 2022}},
  howpublished = {\url{https://martin.ankerl.com/2022/08/27/hashmap-bench-01/}},
  year = {2022}
}

@Misc{hash-table-bench-2,
  author = {Jackson Allan},
  title = {{An Extensive Benchmark of C and C++ Hash Tables}},
  howpublished = {\url{https://jacksonallan.github.io/c_cpp_hash_tables_benchmark/}},
  year = {2024}
}

@Misc{gxhash,
  author =       {Giniaux, Olivier},
  title =        {GxHash: A High-Throughput, Non-Cryptographic Hashing Algorithm
                  Leveraging Modern {CPU} Capabilities},
  year =         2023,
  doi =          {10.5281/ZENODO.8368254},
  publisher =    {Zenodo},
  copyright =    {Open Access}
}

@article{extended-version,
  author       = {Groot Koerkamp, Ragnar and Hermann, Stefan and Sanders, Peter
                  and Walzer, Stefan},
  title        = {Non-minimal k-perfect hashing: Tight lower bounds and an
                  application to fast static hash tables},
  year         = 2026,
  doi          = {10.48550/ARXIV.2607.07257},
  journal      = {arXiv},
  copyright    = {Creative Commons Attribution 4.0 International}
}

\AddToHook{cmd/appendix/before}{%
    \crefalias{section}{appendix}%
    \crefalias{subsection}{appendix}
}
\appendix

\section{Simple Brute Force Suffices for Large Universes}
\label{app:simple-brute-force}

We fill in the missing proof of \cref{thm:brute-force-is-optimal}, meaning for large universes a brute force algorithm need not enforce a balancing condition on the random functions it tries.
\simpleOptimal*
\begin{proof}
    It suffices to show that the probability $\qsimp$ that a uniformly random function is compatible with a given input set satisfies $\qsimp = Θ(\qbal)$, as taking logarithms will turn this into an additive $+𝒪(1)$ discrepancy.
    
    \def\pbal{p_{\mathrm{bal}}}
    \def\psimp{p_{\mathrm{simp}}}
    Let $𝒢$ be the set of all $k$-perfect assignments $g : S → [b]$. Moreover, for $g ∈ 𝒢$ let $\psimp(g)$ (respectively $\pbal(g)$) be the probabilities that a uniformly random (balanced) function matches $g$ when restricted to $S$. We have $\qsimp = \sum_{g ∈ 𝒢} \psimp(g)$ and $\qbal = \sum_{g ∈ 𝒢} \pbal(g)$ so it suffices if we show that $\psimp(g) = Θ(\pbal(g))$ for all $g ∈ 𝒢$. We immediately see $\psimp(g) = b^{-n}$, which means we have to show $\pbal(g) = b^{-n} · Θ(1)$.
    
    To get at $\pbal(g)$, we use the chain rule. The probability that a random balanced function maps the $i$th key $x_i ∈ S$ according to $g$ \emph{conditioned} on having $x₁,…,x_{i-1} ∈ S$ according to $g$ is of the form $\frac cd$ where $c ∈ [\frac un - k, \frac ub]$ is the remaining number of occurrences of $g(x_i)$ and $d ∈ [u-n,u]$ is the total number of unassigned function values. We can bound $\frac cd$ using the assumption $u ≥ n²/α$
    \[
        b^{-1}e^{-Θ(1/n)} ≤ b^{-1}(1-\tfrac{n}{αu}) = b^{-1}(1-\tfrac{kb}{u})
        = \tfrac{u/b -k}{u} ≤ \tfrac cd ≤ \tfrac{u/b}{u-n} ≤ b^{-1}(1+\tfrac{n}{u-n}) ≤ b^{-1}e^{Θ(1/n)}.
    \]
    This implies $\frac{c}{d} = b^{-1}·e^{\pm Θ(1/n)}$ and $n$ such factors as we collect them with the chain rule multiply to $b^{-n}Θ(1)$ as claimed.
\end{proof}

\pagebreak
\section{Numerical approximations of \texorpdfstring{$\bm{S(α,k)}$}{S(α,k)}}
\label{app:sagemath-script}

The following sagemath script (see it in action \href{https://sagecell.sagemath.org/?z=eJyNjc0NgkAQhe8k9DDHmYUIePJCCxagMWaDQCYgSxZMLMSrRXDxgjVsLZbguFCAc5m_973XQA67MHCT9HQjk8wvdDPJniXD7YpHNyvFSaWL0VjWLTJBZSwwcAdWd3WJTZTRicLg83y8DyvsUVaLmfrfw80CV9xdztaYERfLuzjm4CbVxJDGkKW-RD30uigFaE0tOapM8CeieEsQLUefL_v6CYPecjfiHj1K9AW3-laz}{here}) computes the value $S(α,k)$ as defined in \cref{thm:precise-formula}. It is optimised for readability and faithfulness to \cref{thm:precise-formula} and will, in this form fail if $k$ is too large or $α$ too close to $1$.

\begin{algorithm}[H]
    \ttfamily
    $k$ = 8\;
    $α$ = 0.8\;
    $ζ(λ)$ = 1/sum([$λ$**$i$/factorial($i$) for $i$ in range($k$+1)])\;
    $𝔼Z(λ)$ = sum([$i$*$ζ(λ)$*$λ$**$i$/factorial($i$) for $i$ in range($k$+1)])\;
    $λ$ = find\textunderscore root($𝔼Z(x)$ == $α$*$k$, 0, 1000000)\;
    space = log($λ$*e/($α$*$k$),2) + log($ζ(λ)$,2)/($α$*$k$)\;
    print(N(space)) \tcp{prints 0.0829450615181708}
\end{algorithm}

\section{Proof of Theorem \ref{t:asymp}}
\label{app:approx-bounds}

We prove the following theorem from \cref{sec:space-usage-explicit-formulas}.
\approxThm*
Recall the roles of $(X₁,…,X_b)$ and $\qsimp$. We make a few observations before starting the main proof.

\begin{lemma}
    \label{lem:k-is-large}
    \cref{t:asymp} is true for $k = Θ(1)$.
\end{lemma}
\begin{proof}
    If $k = Θ(1)$ then we are necessarily in the case of medium load where a space requirement of $Θ(b) = Θ(n)$ is claimed. This is clearly correct:
    \begin{itemize}
        • $\log₂(e)n = Θ(n)$ bits are sufficient even in the hardest case of $1$-MPHF (a $k$-MPHF can be obtained by constructing a $1$-MPHF and composing it with $i ↦ ⌈i/k⌉$).
        • $Θ(1)$ bits per bin are necessary as $\qsimp ≤ \Pr[X₁ ≤ k]^b$ by \cref{lem:negative-association} and $\Pr[X₁ ≤ k] = 1-Ω(1)$.\qedhere
    \end{itemize}
\end{proof}
We will never explicitly appeal to \cref{lem:k-is-large}. However, keeping in mind that $k$ is large, we allow ourselves to be a bit sloppy with rounding issues and we may appeal to the intuition that the $X_i$ approximately follow a (discretised) normal distribution. This should make the validity of the following lemma apparent without proof.
\begin{lemma}
    \label{lem:basic-binomial}
    Let $N ∈ ℕ$, $p ∈ (0,\frac 12]$ and $X \sim \Bin(N,p)$. Assume $μ := Np = Θ(k)$. Then:
    \begin{itemize}
        • We have $σ² = \Var(X) = Np(1-p) = Θ(k)$.
        • For any $i ∈ ℕ$ with $|i-μ| = 𝒪(\sqrt{k})$ we have $\Pr[X = i] = Θ(1/\sqrt{k})$.
        • $\max_{i ∈ ℕ} \Pr[X = i] = Θ(\Pr[X = ⌊μ⌋]) = Θ(1/\sqrt{k})$.
    \end{itemize}
\end{lemma}
We will require the following Chernoff-type result:
\begin{lemma}
    \label{lem:kullback-leibler-chernoff}
    Let $N ∈ ℕ$, $p < \frac 13$ and $X \sim \Bin(N,p)$. Then
    $\Pr[X ≥ N/2] = \exp(-Θ(N\log \tfrac 1p))$.
\end{lemma}
\begin{proof}
    A standard theorem due to Chernoff and Hoeffding guarantees that
    \[\Pr[X ≥ N(p+ε)] ≤ \exp(-D(p+ε ‖ p)n)\]
    where $D(x ‖ y) = x \ln \frac xy + (1-x) \ln \frac {1-x}{1-y}$ is known as the Kullback-Leibler divergence. Plugging in $ε = \frac 12 - p$ yields the stated result.
\end{proof}
We now have all the tools required for proving \cref{t:asymp}. We consider six cases separately.
    \subparagraph{Lower bound for a case of medium load ($s = Θ(\sqrt{k})$).}
    By \cref{lem:kullback-leibler-chernoff} we have $\Pr[X₁ > k] = 𝒪(1)$ so the lower bound of $Θ(b)$ bits follows from \cref{lem:negative-association} just like in the proof of \cref{lem:k-is-large}.
    \subparagraph{Upper bound for a case of medium load ($s = Θ(\sqrt{k})$).}
    We can handle this case together with the heavily loaded case (considered next). If $s > \sqrt{k}/2$ we may add phantom keys until $s = \sqrt{k}/2$. The resulting upper bound will be $𝒪(b \log(\sqrt{k}/s))$, which simplifies to $𝒪(b)$.
    \subparagraph{Upper bound for the heavily loaded case ($s ≤ \sqrt{k}/2$).}
    We will show $\qsimp = \Pr[\max_{i ∈ [b]} X_i ≤ k] ≥ (Ω(s/\sqrt{k}))^b$, which implies $\log₂ \frac 1\qsimp = 𝒪(b \log(\sqrt{k}/s))$ as desired. In fact, we will show the stronger claim that with $T_i := X₁ + … + X_i - 𝔼[X₁ + … + X_i]$ for $i ∈ [b]$ we have
    \begin{equation}
        \label{eq:perfect-and-never-behind}
        \Pr[∀i ∈ [b]: X_i ≤ k ∧ T_i ≥ 0] = (Ω(s/\sqrt{k}))^b.
    \end{equation}
    Intuitively, we require that when we reveal $X₁,X₂,…,X_b$ one by one, then in step $i$, we fail if $X_i > k$, but we also fail if we ``fall behind'' the number of keys that are expected in the first $i$ bins. Consider the success probability of the $i$th step:
    \begin{align*}
        \Pr[X_i ≤ k ∧ T_i ≥ 0 &\mid ∀j < i: X_j ≤ k ∧ T_j ≥ 0]
        ≥ \Pr[X_i ≤ k ∧ T_i ≥ 0 \mid T_{i-1} = 0]\\
        &= \Pr[X_i ∈ [αk,k] \mid T_{i-1} = 0] = Ω(s/\sqrt{k}).
    \end{align*}
    In the first step, we remove irrelevant information from the conditioning and switch to the ``hardest'' case with $T_{i-1} = 0$ (this ignores a rounding issue) where we have “barely not fallen behind”. The second step uses $T_i = T_{i-1} + X_i - αk$. Finally we apply \cref{lem:basic-binomial} as the conditional distribution of $X_i$ is $\Bin(n-𝔼[X₁+…+X_{i-1}], \frac{1}{b-i+1}) = \Bin(n·\frac{b-i+1}{b}, \frac{1}{b-i+1})$ which has expectation $\frac{n}{b} = αk = Θ(k)$. Now \cref{eq:perfect-and-never-behind} follows from the chain rule.
    \subparagraph{Lower bound for the heavily loaded case ($s \ll \sqrt{k}$).}
    Call a bin $i ∈ [b]$ \emph{good} if $X_i ∈ [k-2s,k]$ and let $G$ be the number of good bins. We begin by showing that the event $\max_{i ∈[b]} X_i ≤ k$ implies $G ≥ b/2$.
    
    Assume, conversely, that $\max_{i ∈[b]} X_i ≤ k$ and $G < b/2$. The non-good bins must all have load less than $k-2s$, the others load at most $k$.
    The total average load is hence less than $\frac{1}{2}(k-2s) + \frac{1}{2}k ≤ k-s = αk$. This is a contradiction as the average load is equal to $αk$.
    
    We now again reveal $X₁,X₂,…,X_b$ one by one. In step $1 ≤ i < b$ the conditional distribution of $X_i$ is binomial and the chance that bin $i$ is good is maximised if the conditional expectation of $X_i$ is around $k-s$. We hence have $\Pr[\text{bin $i$ is good} \mid X₁,…,X_{i-1}] = 𝒪(s/\sqrt{k})$ by \cref{lem:basic-binomial}. A standard coupling argument helps to construct a random variable $G' \sim \Bin(b,p)$ with $p = 𝒪(s/\sqrt{k})$ such that $G ≤ G'$. Note that the bound on $p$ implies $p < \frac 13$. We conclude
    \[
        \qsimp = \Pr[\max X_i ≤ k] ≤ \Pr[G ≥ \tfrac b2] ≤ \Pr[G' ≥ \tfrac b2]
        \stackrel{\text{Lem.\ref{lem:kullback-leibler-chernoff}}}{=}
        \exp(-Ω(b \log \tfrac{1}{p})).
    \]
    Hence $\log₂ \frac 1\qsimp = Ω(b\log \frac 1p) = Ω(b\log(\sqrt{k}/s))$ as required.
    \subparagraph{Upper bound for the lightly loaded case ($s \gg \sqrt{k}$).}
    Note that every bin has $s/\sqrt{k}$ many standard deviations worth of extra space compared to its average load. The probability that a bin overflows is hence $p = \Pr[X_i > k] = \exp(-Ω(s²/k))$ by a standard Chernoff bound.
    
    We will build a $k$-PHF by using a random assignment and fixing these overflowing bins. Firstly, we store the set of overflowing bins using entropy coding. This takes $bp\log \frac 1p + (1-p) \log \frac 1{1-p} = Θ(b p \log \frac 1p) = 𝒪(b\sqrt{p}) = b \exp(-Ω(s²/k))$ bits in expectation, which is within our budget. We now describe how to redistribute keys from overflowing bins to other bins using, in expectation, $𝒪(1)$ bits per overflowing bin. Concretely we store for each overflowing bin:
    \begin{enumerate}
        • The number $D = ⌈(X_i - k)/\sqrt{k}⌉$ of standard deviations by which it overflows. Note that $𝔼[D] = 𝒪(1)$, even conditioned on $X_i > k$. Hence even unary encoding of $D$ is sufficient.
        • A sequence of up to $2D$ seed values (in unary encoding) of hash functions that partition the $X_i$ keys, each hash function splitting off between $\sqrt{k}/2$ and $\sqrt{k}$ keys. We expect each seed to take $𝒪(1)$ bits as we have a constant success probability. The up to $2D$ resulting subsets are called overflow fragments.
        • Imagine each overflow fragment randomly probes the set of all bins until finding a bin with $\sqrt{k}$ spare slots. We store, in unary encoding, the index of the successful probe. As a constant fraction of bins has $\sqrt{k}$ free slots in expectation, and we expect $o(b)$ overflow fragments (start from scratch if this fails), this also takes $𝒪(1)$ bits in expectation.
    \end{enumerate}
    While not elegant, it should be clear that the resulting data structure can serve as a $k$-PHF.
    \subparagraph{Lower bound for the lightly loaded case ($s \gg \sqrt{k}$).}
    We use that $\Pr[X_i > k] = \exp(-Θ(s²/k))$, which follows from local limit theorems\footnote{See e.g.\ \href{https://paperzz.com/doc/7475533/the-moderate-deviations-result?utm_source=chatgpt.com}{lecture notes}
    by Steven Dunbar (see ``key concepts'' $→$ 2).} assuming $n = ω(k⁶)$. We apply \cref{lem:negative-association} to get
    \[
        \qsimp = \Pr[\max X_i ≤ k] ≤ \Pr[X₁ ≤ k]^b = (1-\exp(-Θ(s²/k)))^b.
    \]
    Using that $\log \frac 1{1-ε} = Θ(ε)$ for small $ε$ we get as desired
    \[
        \log₂ \frac 1{\qsimp} ≥ b \log \frac{1}{1-\exp(-Θ(s²/k))} = b\exp(-Θ(s²/k)).
    \]


\section{Engineering the construction algorithm}\label{app:engineering}
Our implementation makes a few specific choices to ensure that construction is
fast.

\parag{Shifting bucket function.}
For each seed that is tried for a bucket, we must compute \(\slot(k_i) =
\h_2(k_i, \seed)\) for all keys in the bucket. If \(\h_2\) is a truly random
function, that means we have a memory access and likely cache miss for each key for each seed in order to
check whether the bin is already full.
In order to reduce this, we only use the high bit of the 8-bit seed for a random hash position, and
the low 7 bits indicate a linear shift relative to that position:
\begin{equation}
\slot(k_i) = \h_2(k_i, \seed) := \h_2'(k_i, \seed \mathbin{\&} 1000\,0000_2) + (\seed \mathbin{\&} 0111\,1111_2).
\end{equation}
This way, seeds \(1\) to \(127\) and $128$ to $255$ map each key to bins that are adjacent in memory.

It is also possible to simply use \emph{all} 8 seed bits to indicate a shift.
While this works mostly fine in practice, it increases the probability of self-collisions in the first
bucket (see below), and usually causes slightly (but occasionally \(2\times\)) more bumping.
Using the high bit for the position effectively enables a
windowed cuckoo-hashing-like effect \cite{walzer2017load}, mitigating
cases where many keys map into the same range of bins.

\parag{Bitmasks.}
Similar to PHast \cite{beling2026phast}, we further speed things up by keeping a bitmask
that indicates which bins still have space. When processing seeds \(0\) and \(128\), we read 128 bits
from the mask starting at the bin for each key, and then intersect these masks.
Then, for each of the next 128 seeds we first check whether the corresponding bit
is \(1\) before processing it further.

\parag{Self-collision checks.}
After the bitmask check passes, there can still be collisions when two keys in a
bucket map to a bin with only one free slot.
To efficiently handle this for consecutive seeds, we compute for seeds \(0\) and
\(128\) a sorted list of pairs \((bin, count)\) indicating how many keys map to each bin. 
For the remaining seeds we simply increment the bin index as needed.
For each seed, we then iterate over the list and simply check whether each bin
has sufficient capacity remaining.

\parag{Computing the score.}
To compute the score, we keep a list \(C=[0, \dots, 0]\) of \(k+1\) counts.
When adding \(c\) keys to a bin with size \(i\), we update \(C_i  \gets C_i+ 1\)
and \(C_{i+c} \gets C_{i+c}- 1\). After taking prefix sums, \(C_i\) then indicates the number of times a key
was added to a bin with \(i\) keys, and we can compute the final score as the
dot product
\(\sum_{i=0}^{k-1} C_i \cdot f(i)\).

\parag{Global seed.}
Typically, the first bucket receives up to \(20\%\) of the keys.
Because of this, the probability that for any fixed seed there is some bin receiving at least \(k+1\)
of these keys is not negligible.
Due to the shifting based on the seed, only 2 independent hash functions are tried, and
it can happen (empirically around $10\%$ of the time) that both lead to self-collisions.
We work around this by making \(h_1\) dependent on a random global seed, so that
we can retry construction until the first bucket does not lead to overflowing bins.

This global seed also serves a secondary purpose to randomize the hashes of keys,
so that the fallback structures for bumped keys have an independent hash function.

\subsection{Consensus}
An alternative way to handle problematic buckets that we considered (but do not use) is \emph{consensus}
\cite{lehmann2025consensus}: we can use as hash a 64-bit seed ending in the 8-bit seed
of the current bucket. This way, the seed of preceding buckets affects the
placement of the seeds in the current bucket, and we can backtrack if no
suitable placement is possible.

While this avoids the need for bumping and the associated cost at query time,
preliminary experiments showed that bumping is more effective at resolving
issues and achieves smaller space overheads. Furthermore, the metadata required
for backtracking somewhat complicates construction, and reading 64-bit seeds
across cache line boundaries would incur additional I/O at query time.

\ifarxiv
  \section{Further results on static hash sets}\label{app:results}
  \Cref{fig:laptop,fig:floyd,fig:diffie}
  show extended results of the kind shown in \cref{set-results} including 
  the time to quere the hashsets in a plain for loop (rather
  than with prefetching) as well as multi-threaded results.
  
  \begin{figure}[h]
    \centering
    \makebox[\linewidth]{
    \includegraphics[width=1.2\linewidth]{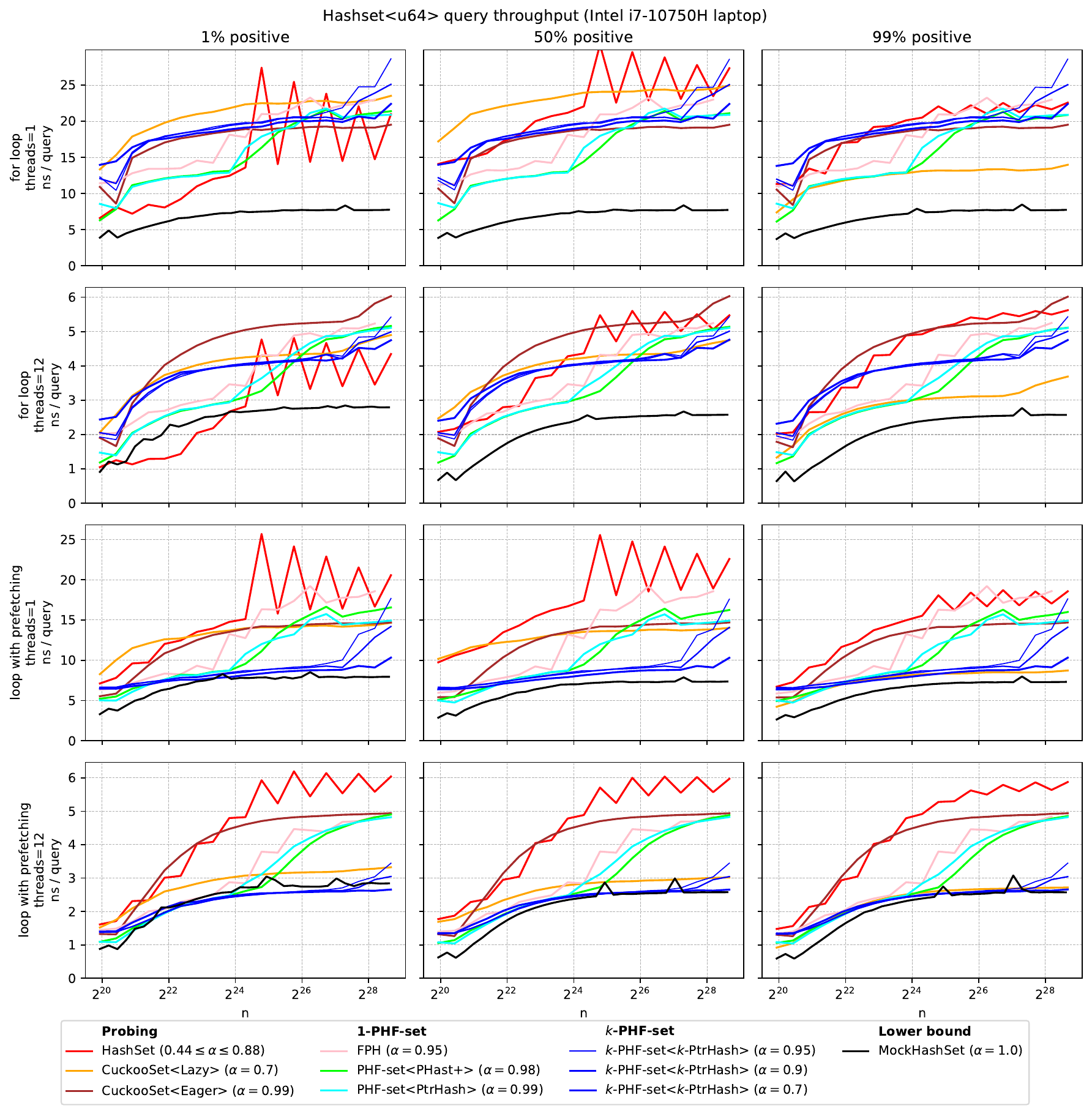}
    }
    \caption{Results on 1 and 12 threads for the throughput of a for loop and loop
      with prefetching on an Intel Skylake laptop CPU.}
    \label{fig:laptop}
  \end{figure} 
  \begin{figure}[h]
    \centering
    \makebox[\linewidth]{
    \includegraphics[width=1.2\linewidth]{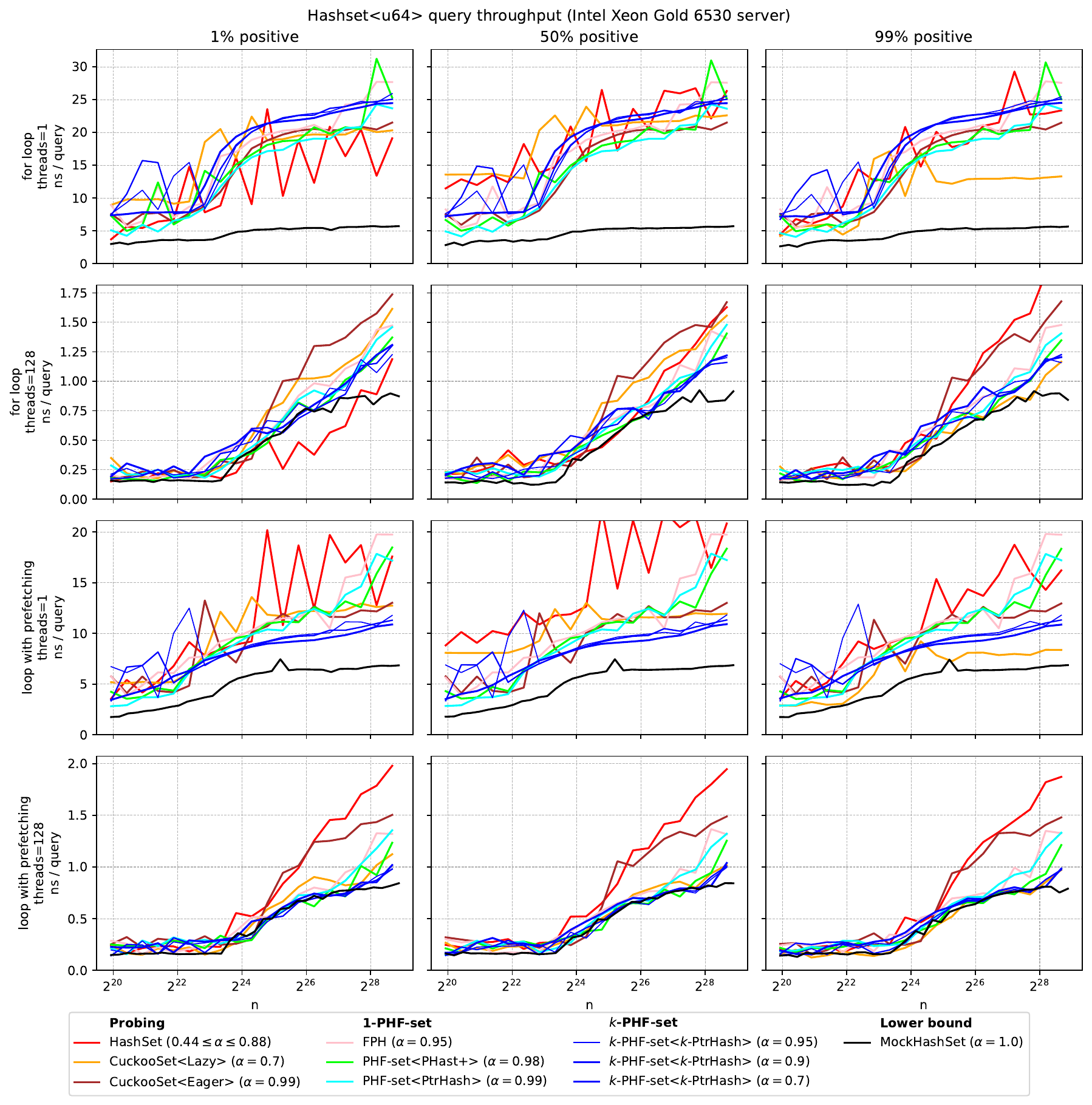}
    }
    \caption{Results on 1 and 128 threads for the throughput of a for loop and loop
      with prefetching on an Intel Xeon server CPU.}
    \label{fig:floyd}
  \end{figure} 
  \begin{figure}[h]
    \centering
    \makebox[\linewidth]{
    \includegraphics[width=1.2\linewidth]{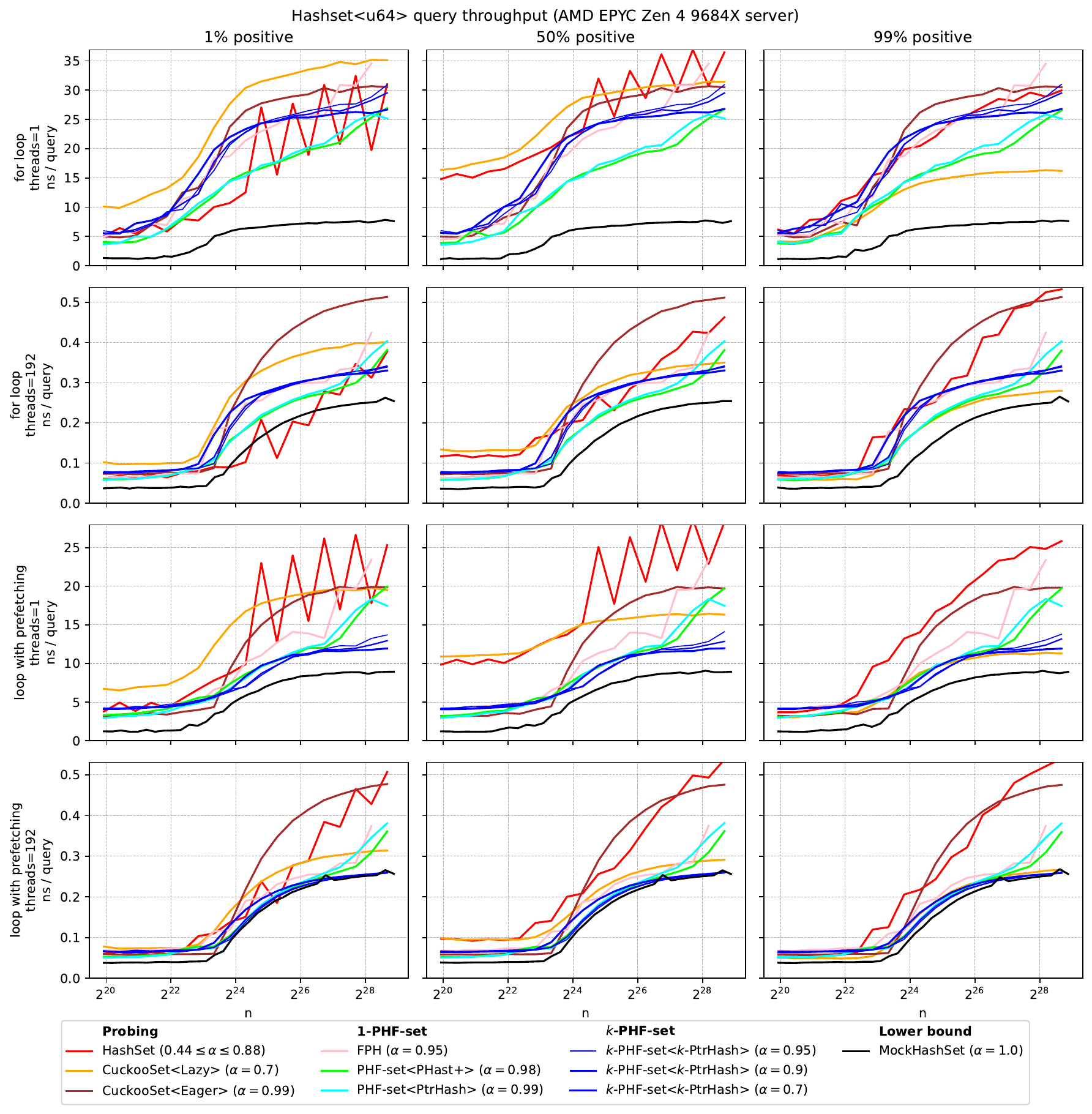}
    }
    \caption{Results on 1 and 192 threads for the throughput of a for loop and loop
    with prefetching on an AMC EPYC server CPU.}
    \label{fig:diffie}
  \end{figure}
\fi

\end{document}